\documentclass[11pt,letterpaper,onecolumn,accepted=2021-09-05]{quantumarticle}
\pdfoutput=1

\usepackage[margin=1in]{geometry}
\usepackage{amsmath,amsthm,amsfonts,amssymb}
\usepackage{mathtools}
\mathtoolsset{centercolon}
\usepackage[T1]{fontenc}
\usepackage{textcomp}
\usepackage{tikz}

\usepackage{enumitem}
\usepackage{microtype}
\usepackage{xspace}
\usepackage{braket} 

\usepackage[pagebackref]{hyperref} 
\hypersetup{
   pdftitle={}, 
   pdfauthor={}, 
   colorlinks=true, 
   linkcolor=blue, 
   citecolor=blue, 
   urlcolor=blue 
}

\usepackage[noabbrev,capitalize,nameinlink]{cleveref}
\usepackage{thm-restate}

\usepackage{wrapfig}
\usepackage{lipsum}
\usepackage{tocloft}
\usepackage{algorithm}
\usepackage{algpseudocode}
\newcommand*\Let[2]{\State #1 $\gets$ #2}
\algrenewcommand{\algorithmiccomment}[1]{\hfill $\triangleright$ #1}

\renewcommand{\backref}[1]{}

\renewcommand{\backrefalt}[4]{%
\ifcase #1 %

\or 
[p.\ #2]%
\else 
[pp.\ #2]%
\fi}

\makeatletter
\renewcommand{\paragraph}{%
 \@startsection{paragraph}{4}%
 {\z@}{2.25ex \@plus .5ex \@minus .3ex}{-1em}%
 {\normalfont\normalsize\bfseries}%
}
\makeatother

\makeatletter
\newcommand{\para}{%
 \@startsection{paragraph}{4}%
 {\z@}{2ex \@plus 3.3ex \@minus .2ex}{-1em}%
 {\normalfont\normalsize\bfseries}%
}
\makeatother

\newtheorem{theorem}{Theorem}
\newtheorem{claim}[theorem]{Claim}
\newtheorem{lemma}[theorem]{Lemma}
\newtheorem{proposition}[theorem]{Proposition}
\newtheorem{corollary}[theorem]{Corollary}

\theoremstyle{definition}
\newtheorem{definition}[theorem]{Definition}
\newtheorem{openproblem}{Open Problem}

\newcommand{\eq}[1]{\hyperref[eq:#1]{(\ref*{eq:#1})}}
\renewcommand{\sec}[1]{\hyperref[sec:#1]{Section~\ref*{sec:#1}}}
\newcommand{\thm}[1]{\hyperref[thm:#1]{Theorem~\ref*{thm:#1}}}
\newcommand{\conj}[1]{\hyperref[thm:#1]{Conj~\ref*{thm:#1}}}
\newcommand{\lem}[1]{\hyperref[lem:#1]{Lemma~\ref*{lem:#1}}}
\newcommand{\prop}[1]{\hyperref[prop:#1]{Proposition~\ref*{prop:#1}}}
\newcommand{\cor}[1]{\hyperref[cor:#1]{Corollary~\ref*{cor:#1}}}
\newcommand{\fig}[1]{\hyperref[fig:#1]{Figure~\ref*{fig:#1}}}
\newcommand{\tab}[1]{\hyperref[tab:#1]{Table~\ref*{tab:#1}}}
\newcommand{\alg}[1]{\hyperref[alg:#1]{Algorithm~\ref*{alg:#1}}}
\newcommand{\app}[1]{\hyperref[app:#1]{Appendix~\ref*{app:#1}}}
\newcommand{\defn}[1]{\hyperref[def:#1]{Definition~\ref*{def:#1}}}

\newcommand{\C}{{\mathcal{C}}}
\newcommand{\R}{{\mathbb{R}}}

\newcommand{\D}{\mathcal{D}}

\renewcommand{\(}{\left(}
\renewcommand{\)}{\right)}

\newcommand{\be}{\begin{equation}}
\newcommand{\ee}{\end{equation}}
\def\ba#1\ea{\begin{align}#1\end{align}}

\newcommand{\AMAJ}{\mathsf{AMAJ}}
\newcommand{\IP}{\mathsf{IP}}
\newcommand{\B}{\{0,1\}}
\newcommand{\Bo}{\{0,1\}}
\newcommand{\Bf}{\{-1,1\}}

\newcommand{\F}{\mathcal{F}}
\newcommand{\eps}{\varepsilon}

\DeclareMathOperator{\adeg}{\widetilde{\deg}}

\DeclareMathOperator{\AND}{\mathsf{AND}}

\DeclareMathOperator{\OR}{\mathsf{OR}}
\DeclareMathOperator{\NOT}{\mathsf{NOT}}

\DeclareMathOperator{\PARITY}{\mathsf{PARITY}}
\DeclareMathOperator{\polylog}{polylog}
\DeclareMathOperator{\poly}{poly}
\DeclareMathOperator{\mon}{\mu}

\DeclareMathOperator{\amu}{\widetilde{\mu}}

\newcommand{\tO}{\widetilde{O}}
\newcommand{\tOmega}{\widetilde{\Omega}}

\newcommand{\AC}{{$\mathsf{AC}^0$}\xspace} 
\newcommand{\ACmodtwo}{{$\mathsf{AC}^0 \circ \oplus$}\xspace} 
\newcommand{\ACd}{{$\mathsf{AC}^0_d$}\xspace} 
 
\newcommand{\LC}{{$\mathsf{LC}^{0}$}\xspace} 
\newcommand{\LCd}{\ensuremath{\mathsf{LC}^{0}_d}\xspace} 
\newcommand{\LCdmone}{{$\mathsf{LC}^0_{d-1}$}\xspace} 
\newcommand{\LCone}{$\mathsf{LC}^0_1$\xspace}

\newcommand{\LCthree}{$\mathsf{LC}^0_3$\xspace}

\newcommand\blfootnote[1]{%
  \begingroup
  \renewcommand\thefootnote{}\footnote{#1}%
  \addtocounter{footnote}{-1}%
  \endgroup
}

\begin{document}

\title{Quantum algorithms and approximating polynomials for composed functions with shared inputs}

\author{Mark Bun}
\affiliation{Boston University \textsf{mbun@bu.edu}}
\author{Robin Kothari}
\affiliation{Microsoft Quantum \textsf{robin.kothari@microsoft.com}}
\author{Justin Thaler}
\affiliation{Georgetown University \textsf{justin.thaler@georgetown.edu}}

\date{}
\maketitle

\begin{abstract}
We give new quantum algorithms for evaluating composed functions whose inputs may be shared between bottom-level gates. Let $f$ be an $m$-bit Boolean function and consider an $n$-bit function $F$ obtained by applying $f$ to conjunctions of possibly overlapping subsets of $n$ variables. If $f$ has quantum query complexity $Q(f)$, we give an algorithm for evaluating $F$ using $\tilde{O}(\sqrt{Q(f) \cdot n})$ quantum queries. This improves on the bound of $O(Q(f) \cdot \sqrt{n})$ that follows by treating each conjunction independently, and our bound is tight for worst-case choices of $f$. Using completely different techniques, we prove a similar tight composition theorem for the approximate degree of $f$.

By recursively applying our composition theorems, we obtain a nearly optimal $\tilde{O}(n^{1-2^{-d}})$ upper bound on the quantum query complexity and approximate degree of linear-size depth-$d$ \AC circuits. As a consequence, such circuits can be PAC learned in subexponential time, even in the challenging agnostic setting. 
Prior to our work, a subexponential-time algorithm was not known even for linear-size depth-3 \AC circuits.

As an additional consequence, we show that  \ACmodtwo\ circuits of depth $d+1$ require size
 $\tOmega(n^{1/(1- 2^{-d})}) \geq \omega(n^{1+ 2^{-d}} )$ to compute the Inner Product function
 even on average. The previous best size lower bound was $\Omega(n^{1+4^{-(d+1)}})$ and only 
 held in the worst case (Cheraghchi et al., JCSS 2018).
 \blfootnote{A preliminary version of this manuscript appeared in ACM-SIAM Symposium on Discrete Algorithms (SODA), 2019~\cite{BKT19}. That version did not contain the lower bound for \ACmodtwo\ circuits computing the Inner Product function.}
\end{abstract}

\clearpage

\section{Introduction}
\label{sec:intro}

In the query, or black-box, model of computation, an algorithm aims to evaluate a known Boolean function $f : \B^n \to \B$ on an unknown input $x \in \B^n$ by reading as few bits of $x$ as possible. One of the most basic questions one can ask about query complexity, or indeed any complexity measure of Boolean functions, is how it behaves under \emph{composition}. Namely, given functions $f$ and $g$, and a method of combining these functions to produce a new function $h$, how does the query complexity of $h$ depend on the complexities of the constituent functions $f$ and $g$?

The simplest method for combining functions is \emph{block composition}, where the inputs to $f$ are obtained by applying the function $g$ to independent sets of variables. That is, if $f : \B^m \to \B$ and $g : \B^k \to \B$, then the block composition $(f \circ g): \B^{m \cdot k} \to \B$ is defined by $(f \circ g)(x_1, \dots, x_m) = f(g(x_1), \dots, g(x_m))$ where each $x_i$ is a $k$-bit string. In most reasonable models of computation, one can evaluate $f \circ g$ by running an algorithm for $f$, and using an algorithm for $g$ to compute the inputs to $f$ as needed. Thus, the query complexity of $f \circ g$ is at most the product of the complexities of $f$ and $g$.\footnote{In some ``reasonable models,'' such as those with bounded error, one must take care to ensure that errors in computing each copy of $g$ do not propagate, but we elide these issues for this introduction. Addressing this concern typically adds at most a logarithmic overhead.}

For many query models, including those capturing deterministic and quantum computation, this is known to be tight. In particular, letting $Q(f)$ denote the bounded-error quantum query complexity of a function $f$, it is known that $Q(f \circ g) = \Theta(Q(f) \cdot Q(g))$ for all Boolean functions $f$ and $g$~\cite{negativeweights, Rei11}. This result has the flavor of a direct sum theorem: When computing many copies of the function $g$ (in this case, as many as are needed to generate the necessary inputs to $f$), one cannot do better than just computing each copy independently.

\subsection{Quantum algorithms for shared-input compositions}

While we have a complete understanding of the behavior of quantum query complexity under block composition, little is known for more general compositions. What is the quantum query complexity of a composed function where inputs to $f$ are generated by applying $g$ to overlapping sets of variables? We call these more general compositions \emph{shared-input compositions}. Not only does answering this question serve as a natural next step for improving our understanding of quantum query complexity, but it may lead to more unified algorithms and lower bounds for specific functions of interest in quantum computing. Many of the functions  that have played an influential role in the study of quantum query complexity can be naturally expressed as compositions of simple functions with shared inputs, including $k$-distinctness, $k$-sum, surjectivity, triangle finding, and graph collision.

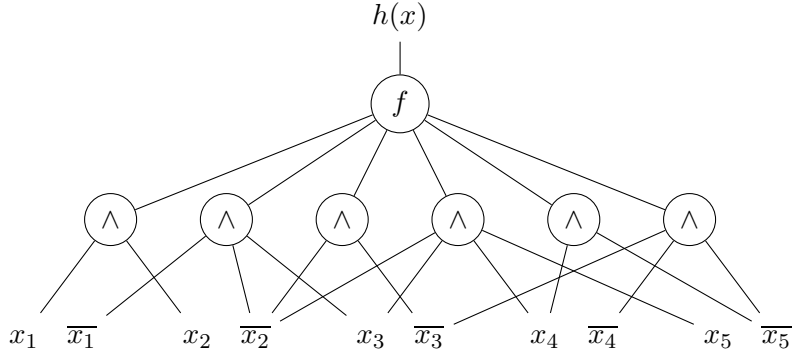
\begin{figure}[ht]
\centering
\begin{tikzpicture}[x=4em,y=4em,scale=1]

\node at (2,2.75) [] (h) {$h(x)$};

\node at (2,2) [shape=circle, draw] (f) {$f$};

\node at (-1.25,0) [] (x1) {\vphantom{I}$x_1$};
\node at (-0.75,0) [] (notx1) {\vphantom{I}$\overline{x_1}$};
\node at (0.25,0) [] (x2) {\vphantom{I}$x_2$};
\node at (0.75,0) [] (notx2) {\vphantom{I}$\overline{x_2}$};
\node at (1.75,0) [] (x3) {\vphantom{I}$x_3$};
\node at (2.25,0) [] (notx3) {\vphantom{I}$\overline{x_3}$};
\node at (3.25,0) [] (x4) {\vphantom{I}$x_4$};
\node at (3.75,0) [] (notx4) {\vphantom{I}$\overline{x_4}$};
\node at (4.75,0) [] (x5) {\vphantom{I}$x_5$};
\node at (5.25,0) [] (notx5) {\vphantom{I}$\overline{x_5}$};
\node at (-0.5,1) [shape=circle, draw] (g1) {$\wedge$};
\node at (0.5,1) [shape=circle, draw] (g2) {$\wedge$};
\node at (1.5,1) [shape=circle, draw] (g3) {$\wedge$};
\node at (2.5,1) [shape=circle, draw] (g4) {$\wedge$};
\node at (3.5,1) [shape=circle, draw] (g5) {$\wedge$};
\node at (4.5,1) [shape=circle, draw] (g6) {$\wedge$};

\draw (x1) -- (g1);
\draw (notx1) -- (g2);
\draw (x2) -- (g1);
\draw (notx2) -- (g2);
\draw (notx2) -- (g3);
\draw (notx2) -- (g4);
\draw (notx3) -- (g3);
\draw (x3) -- (g2);
\draw (x3) -- (g4);
\draw (notx3) -- (g6);
\draw (notx4) -- (g6);
\draw (x4) -- (g4);
\draw (x4) -- (g5);
\draw (x5) -- (g4);
\draw (notx5) -- (g5);
\draw (notx5) -- (g6);

\draw (g1) -- (f);
\draw (g2) -- (f);
\draw (g3) -- (f);
\draw (g4) -- (f);
\draw (g5) -- (f);
\draw (g6) -- (f);
\draw (h) -- (f);

\end{tikzpicture}
\caption{\label{fig:composition}A depth-2 circuit $h:\B^5\to\B$ where the top gate is a function $f:\B^6 \to \B$ and the bottom level gates are  $\AND$ gates on a subset of the input bits and their negations. More generally, we consider $h:\B^n \to \B$, with top gate $f:\B^m \to \B$.}
\end{figure}

In this work, we study shared-input compositions between an arbitrary function $f$ and the function $g = \AND$. If $f : \B^m \to \B$, then we let $h : \B^n \to \B$ be any function obtained by generating each input to $f$ as an $\AND$ over some subset of (possibly negated) variables from $x_1, \dots, x_n$, as depicted in \fig{composition}.

Of course, one can compute the function $h$ by ignoring the fact that the $\AND$ gates depend on shared inputs, and instead regard each gate as depending on its own set of copies of the input variables. Using the quantum query upper bound for block compositions, together with the fact that $Q(\AND_n) = \Theta(\sqrt{n})$~\cite{Gro96,BBBV97}, one obtains 
\begin{equation}
Q(h) = O(Q(f) \cdot Q(\AND_n)) = O(Q(f) \cdot \sqrt{n}).
\end{equation}
Observe that this bound on $Q(h)$ is non-trivial only if $Q(f) \ll \sqrt{n}$.
A priori, one may conjecture that this bound is tight in the worst case for shared-input compositions. After all, if the variables overlap in some completely arbitrary way with no structure, it is unclear from the perspective of an algorithm designer how to use the values of already-computed $\AND$ gates to reduce the number of queries needed to compute further $\AND$ gates. It might even be the case that every pair of $\AND$ gates shares very few common input bits, suggesting that evaluating one $\AND$ gate yields almost no information about the output of any other $\AND$ gate. This intuition even suggests a path for proving a matching lower bound: Using a random wiring pattern, combinatorial designs, etc., construct the set of inputs to each $\AND$ gate so that evaluating any particular gate leaks almost no useful information that could be helpful in evaluating the other $\AND$ gates. 

In this work, we show that this intuition is wrong: the overlapping structure of the $\AND$
gates can \emph{always} be exploited algorithmically (so long as $Q(f) \ll n$). 

\para{Results.} Our main result shows that a shared-input composition between a function $f$ and the $\AND$ function always has substantially lower quantum query complexity than the block composition $f \circ \AND_n$. Specifically, instead of having quantum query complexity which is the product $Q(f) \cdot \sqrt{n}$, a shared-input composition has quantum query complexity which is, up to logarithmic factors, the geometric mean $\sqrt{Q(f) \cdot n}$ between $Q(f)$ and the number of input variables $n$. This bound is nontrivial whenever $Q(f)$ is significantly smaller than $n$. 

\begin{restatable}{theorem}{quantum}
\label{thm:quantum}
Let $h:\B^n \to \B$ be computed by a depth-2 circuit where the top gate is a function $f:\B^m \to \B$ and the bottom level gates are  $\AND$ gates on a subset of the input bits and their negations (as depicted in \fig{composition}). Then we have
\begin{equation}
Q(h) = O\Bigl(\sqrt{Q(f) \cdot n} \cdot \log^2(mn)\Bigr).
\end{equation}
\end{restatable}

Note that \thm{quantum} is nearly tight for every possible value of $Q(f) \in [n]$.\footnote{ \thm{quantum} is not tight for every function $f$, of course. For example if $f$ is an $\AND$ on many inputs, the composed function will have quantum query complexity $O(\sqrt{n})$ but the upper bound of \thm{quantum} can be larger than this.} For a parameter $t \le n$, consider the block composition (i.e., the composition with disjoint inputs) $\PARITY_t \circ \AND_{n / t}$. Since $Q(\PARITY_t) = \lceil t / 2 \rceil$~\cite{beals}, this function has quantum query complexity
\begin{equation}
Q\(\PARITY_t \circ \AND_{n / t}\)=\Theta\left(t \cdot \sqrt{n/t}\right) = \Theta\left(\sqrt{Q(\PARITY_t) \cdot n}\right),
\end{equation}
matching the upper bound provided by \thm{quantum} up to log factors. This shows that \thm{quantum} cannot be significantly improved in general.

The proof of \thm{quantum} makes use of an optimal quantum algorithm for computing $f$ and Grover's search algorithm for evaluating $\AND$ gates. 
Surprisingly, it uses no other tools from quantum computing. 
The core of the argument is entirely classical, relying on a recursive gate and wire-elimination argument for evaluating $\AND$ gates with overlapping inputs. 

At a high level, the algorithm in \thm{quantum} works as follows. The overall goal is to  query enough input bits such that the resulting circuit is simple enough to apply the composition upper bound $Q(f \circ g) = O(Q(f)Q(g))$. To apply this upper bound and obtain the claimed upper bound in \thm{quantum}, we require $Q(g)$ to be $O(\sqrt{n/Q(f)})$. Since $g$ is just an $\AND$ gate on some subset of inputs, this means we want the fan-in of each $\AND$ gate in our circuit to be $O(n/Q(f))$. If we call $\AND$ gates with fan-in $\omega(n/Q(f))$ ``high fan-in'' gates, then the goal is to eliminate all high fan-in gates. Our algorithm achieves this by judiciously querying input bits that would eliminate a large number of high fan-in gates if they were set to 0. 

Besides the line of work on the quantum query complexity of block compositions, our result is also closely related to work of Childs, Kimmel, and Kothari~\cite{CKK12} on read-many formulas. Childs et al. showed that any formula on $n$ inputs consisting of $G$ gates from the de Morgan basis $\{\AND, \OR, \NOT\}$ can be evaluated using $O(G^{1/4} \cdot \sqrt{n})$ quantum queries. In the special case of DNF formulas, our result coincides with theirs by taking the top function $f$ to be the $\OR$ function. However, even in this special case, the result of Childs et al. makes critical use of the top function being $\OR$. Specifically, their result uses the fact that the quantum query complexity of the $\OR$ function is the square root of its formula size. Our result, on the other hand, applies without making any assumptions on the top function $f$. This level of generality is needed when using \thm{quantum} to understand \emph{circuits} (rather than just formulas) of depth 3 and higher, as discussed in \sec{intro-lc}.

\subsection{Approximate degree of shared-input compositions}

We also study shared-input compositions under the related notion of approximate degree. For a Boolean function $f : \B^n \to \B$, an \emph{$\eps$-approximating polynomial} for $f$ is a real polynomial $p : \B^n \to \R$ such that $|p(x) - f(x)| \le \eps$ for all $x \in \B^n$. The \emph{$\eps$-approximate degree} of $f$, denoted $\deg_\eps(f)$, is the least degree among all $\eps$-approximating polynomials for $f$. We use the term \emph{approximate degree} without qualification to refer to choice $\eps = 1/3$, and denote it $\adeg(f) = \deg_{1/3}(f)$.

A fundamental observation due to Beals et al.~\cite{beals} is that any $T$-query quantum algorithm for computing a function $f$ implicitly defines a degree-$2T$ approximating polynomial for $f$. Thus, $\adeg(f) \le 2Q(f)$. This relationship has led to a number of successes in proving quantum query complexity lower bounds via approximate degree lower bounds, constituting a technique known as the polynomial method in quantum computing. Conversely, quantum algorithms are powerful tools for establishing the existence of low-degree approximating polynomials that are needed in other applications to theoretical computer science. For example, the deep result that every de Morgan formula of size $s$ has quantum query complexity, and hence approximate degree, $O(\sqrt{s})$~\cite{FGG08,CCJY07,ACRSZ10,Rei11} underlies the fastest known algorithm for agnostically learning formulas~\cite{kkms, Rei11} (See \sec{intro-learning} and \sec{agnostic} for details on this application). It has also played a major role in the proofs of the strongest formula and graph complexity lower bounds for explicit functions~\cite{tal1}.

\para{Results.} We complement our result on the quantum query complexity of shared-input compositions with an analogous result for approximate degree. 
\begin{restatable}{theorem}{composition}
\label{thm:composition}
Let $h:\B^n\to \B$ be computed by a depth-2 circuit where the top gate is a function $f:\B^m \to \B$ and the bottom level gates are  $\AND$ gates on a subset of the input bits and their negations (as depicted in \fig{composition}). Then 
\begin{equation}
\deg_\eps(h) =O\(\sqrt{\deg_\eps(f)\cdot n \cdot \log m} + \sqrt{n\log(1/\eps)}\).\label{eq:composition}
\end{equation}
In particular, $\adeg(h) = O\(\sqrt{\adeg(f) \cdot n\log m}\)$.
\end{restatable}

Note that our result for approximate degree is incomparable with \thm{quantum}, even for bounded error, since both sides of the equation include the complexity measure under consideration.

Like \thm{quantum}, \thm{composition} can be shown to be tight by considering the block composition of $\PARITY$ with $\AND$, since $\adeg(\PARITY_{t}\circ \AND_{n/t})=\Theta\(\sqrt{\adeg(\PARITY_t)\cdot n}\)$ \cite{sherstovrobust, comm5}.

Our proof of \thm{composition} abstracts and generalizes a technique introduced by Sherstov~\cite{algorithmicpolys}, who very recently proved an $O(n^{3/4})$ upper bound on the approximate degree of an important depth-3 circuit of nearly quadratic size called Surjectivity~\cite{algorithmicpolys}. Despite the similarity between \thm{composition} and \thm{quantum}, and the close connection between approximating polynomials and quantum algorithms, the proof of \thm{composition} is completely different from \Cref{thm:quantum}, making crucial use of properties of polynomials that do not hold for 
quantum algorithms.\footnote{Any analysis capable of yielding a sublinear upper bound
on the approximate degree of Surjectivity requires moving beyond quantum algorithms, as its quantum query complexity is known to be ${\Omega}(n)$~\cite{beame,sherstov15}.} In our opinion, this feature of the proof of \thm{composition} makes \thm{quantum} for quantum algorithms even more surprising.

We remark that a different proof of the $O(n^{3/4})$ upper bound for the approximate degree of Surjectivity was discovered in \cite{BKT18},
who also showed a matching lower bound. It is also possible to prove \thm{composition} by generalizing the techniques developed in that work, but the techniques of \cite{algorithmicpolys} lead to a shorter and cleaner analysis.

\subsection{Application: Evaluating and approximating linear-size \AC circuits} \label{sec:intro-lc}

The circuit class \AC consists of constant-depth, polynomial-size circuits over the de Morgan basis $\{\AND, \OR, \NOT\}$ with unbounded fan-in gates. The full class \AC is known to contain very hard functions from the standpoint of both quantum query complexity and approximate degree. The aforementioned Surjectivity function is in depth-3 \AC and has quantum query complexity $\Omega(n)$~\cite{beame,sherstov15}, while for every positive constant $\delta > 0$, there exists a depth-$O(\log(1/\delta))$ \AC circuit with approximate degree $\Omega(n^{1-\delta})$~\cite{BT17}.

Nevertheless, \AC contains a number of interesting subclasses for which nontrivial quantum query and approximate degree upper bounds might still hold. Here, we discuss applications of our composition theorem to understanding the subclass \LC, consisting of \AC circuits of linear size.

The class \LC is one of the most interesting subclasses of \AC. It has been studied by many authors in various complexity-theoretic contexts, ranging from logical characterizations~\cite{KLPT06} to faster-than-brute-force satisfiability algorithms~\cite{calabroetal,santhanam2012limits}.
\LC turns out to be a surprisingly powerful class. For example, the $k$-threshold function that asks if the input has Hamming weight greater than $k$ is clearly in \AC for constant $k$, by computing the $\OR$ of all $\binom{n}{k}$ possible certificates. But this yields a circuit of size $O(n^k)$, which one might conjecture is optimal. However, it turns out that $k$-threshold is in \LC even when $k$ is as large as $\polylog(n)$~\cite{RW91}.
Another surprising fact is that every regular language in \AC can be computed by an \AC circuit of almost linear size (e.g., size $O(n\log^* n)$ suffices)~\cite{Kou09}.
  
By recursively applying \thm{quantum}, we obtain the following sublinear upper bound on the quantum query complexity of depth-$d$ \LC circuits, denoted by $\LCd$:

\begin{restatable}{theorem}{quantumLC} \label{thm:quantumLC}
For all constants $d\geq 0$ and all functions $h:\B^n \to \B$ in $\LCd$, we have $Q(h) = \tO(n^{1-{2^{-d}}})$.
\end{restatable}

Our upper bound is nearly tight for every depth $d$, as shown in~\cite{CKK12}.

\begin{restatable}[Childs, Kimmel, and Kothari]{theorem}{quantumlower}
\label{thm:quantumlower}
For all constants $d\geq 0$, there exists a function $h:\B^n \to \B$ in $\LCd$ with $Q(h) \geq n^{1-{2^{-\Omega(d)}}}$.
\end{restatable}

By recursively applying \thm{composition}, we obtain a similar 
sublinear upper bound for the $\eps$-approximate degree of $\LCd$, even for subconstant values of $\eps$.

\begin{restatable}{theorem}{approxdegLC}\label{thm:approxdegLC}
	For all constant $d\geq 0$, and any $\eps > 0$, and all functions $h:\B^n \to \B$ in $\LCd$, we have 
\begin{equation}
\deg_\eps(h) = \tO\Bigl(n^{1-2^{-d}}\log^{2^{-d}}(1/\eps)\Bigr).
\end{equation}
\end{restatable}

For constant $\epsilon$, we prove a lower bound of the same form with quadratically worse dependence on the depth $d$.

\begin{restatable}{theorem}{lowerbound}
\label{thm:lower}
For all constants $d\geq 0$, there exists a function $h:\B^n\to\B$ in \LCd with $\adeg(h) \geq  n^{1-2^{-\Omega(\sqrt{d})}}$.
\end{restatable}

A lower bound of $\adeg(h) = n^{1-2^{-\Omega(d)}}$ was already known for general \AC functions $f$ \cite{BT17, BKT18}, but the \AC circuits constructed in these prior works are not of linear size.  Previously, for any $\ell \geq 1$, \cite{BKT18} exhibited a circuit $C \colon \B^n \to \B$ of depth at most $3\ell$, size at most $n^2$, and approximate degree $\adeg(C) \geq  \tOmega(n^{1-2^{-\ell}})$. We show how to transform this quadratic-size circuit $C$
into a linear-size circuit $C$ of depth roughly $\ell^2$, whose approximate degree is close to that of $C$.  Our transformation adapts that of~\cite{CKK12}, but requires a more intricate construction and analysis. This is because, unlike quantum query complexity, approximate degree is not known to increase multiplicatively under block composition.

\subsection{Application: Agnostically learning linear-size \AC circuits} \label{sec:intro-learning}

The challenging agnostic model~\cite{KSS:1994} of computational learning theory captures the task of binary classification in the presence of adversarial noise. In this model, a learning algorithm is given a sequence of labeled examples of the form $(x, b) \in \B^n \times \B$ drawn from an unknown distribution $\D$. The goal of the algorithm is to learn a hypothesis $h : \B^n \to \B$ which does ``almost as well'' at predicting the labels of new examples drawn from $\D$ as does the the best classifier from a known concept class $\C$. Specifically, let the Boolean loss of a hypothesis $h$ be $\mathsf{err}_{\D}(h) = \Pr_{(x, b) \sim \D}[h(x) \ne b]$. For a given accuracy parameter $\eps$, the goal of the learner is to produce a hypothesis $h$ such that $\mathsf{err}_{\D}(h) \le \min_{c \in \C} \mathsf{err}_{\D}(c) + \eps$.

Very few concept classes $\C$ are known to be agnostically learnable, even in subexponential time. For example, the best known algorithm for agnostically learning disjunctions runs in time $2^{\tilde{O}(\sqrt{n})}$~\cite{kkms}.\footnote{Throughout this manuscript, $\tilde{O}$ and $\tilde{\Omega}$ notation hides factors polylogarithmic in the input size $n$.} Moreover, several hardness results are known. Proper agnostic learning of disjunctions (where the output hypothesis itself must be a disjunction) is NP-hard~\cite{KSS:1994}. Even improper agnostic learning of disjunctions is at least as hard as PAC learning DNF~\cite{lbw},
which is a longstanding open question in learning theory. \label{sec:introst}

The best known general result for more expressive classes of circuits is that all de Morgan \emph{formulas} of size $s$ can be learned in time $2^{\tilde{O}(\sqrt{s})}$~\cite{kkms, Rei11} (\sec{detailed} contains a detailed overview of prior work on agnostic and PAC learning). Both of the aforementioned results make use of the well-known \emph{linear regression} framework of~\cite{kkms} for agnostic learning. This algorithm works whenever there is a ``small'' set of ``features'' $\F$ (where each feature is a function mapping $\B^n$ to $\R$) such that each concept in the concept class $\C$ can be approximated to error $\eps$ in the $\ell_\infty$ norm by a linear combination of features in $\F$. (See \sec{agnostic} for details.) 
If every function in a concept class $\C$ has approximate degree at most $d$, then  one obtains an agnostic learning algorithm for $\C$ with running time $2^{\tilde{O}(d)}$ by taking $\F$ to be the set of all monomials of degree at most $d$. Applying this algorithm using the approximate degree upper bound of \thm{approxdegLC} yields a subexponential time algorithm for agnostically learning $\LCd$.

\begin{restatable}{theorem}{learningresult}
    \label{thm:learning}
    The concept class of $n$-bit functions computed by \LC circuits of
    depth $d$ can be learned in the distribution-free agnostic PAC model in time
    $2^{\tO(n^{1-2^{-d}})}$.
More generally, size-$s$ \ACd circuits can be learned in time
    $2^{\tO(\sqrt{n} s^{1/2 - 2^{-d}})}$.
\end{restatable}

Prior to our work, no subexponential time algorithm was known even for agnostically learning \LCthree. Moreover, since our upper bound on the approximate degree of \LC circuits is nearly tight, new techniques will be needed to significantly surpass our results, and in particular, learn \emph{all} of \LC in subexponential time. (Note that standard techniques~\cite{patmat} automatically generalize the lower bound of \thm{lower} from the feature set of low-degree monomials to \emph{arbitrary feature sets}. See \Cref{sec:agnostics} for details.)

\subsection{Application: New Circuit Lower Bounds}
An important frontier problem in circuit complexity is to show
that the well-known Inner Product function cannot be computed by \ACmodtwo circuits of polynomial size.
Here, \ACmodtwo\ refers to \AC circuits augmented with a layer of parity gates at the bottom (i.e., 
closest to the inputs). Servedio and Viola \cite{servedio2012special} identified this open problem as a first step
toward proving matrix rigidity lower bounds, itself a notorious open problem in complexity theory,
and Akavia et al. \cite{akavia} connected the problem to the goal of constructing highly efficient pseudorandom generators.\footnote{Superpolynomial
lower bounds are known for \ACmodtwo\ circuits computing the Majority function \cite{razborov} (in fact, even for $\mathsf{AC}^0[2]$ circuits, which are \AC
circuits augmented with parity gates at any layer). However, these 
techniques do not apply to the Inner Product function, which does have small $\mathsf{AC}^0[2]$ circuits.}
Average-case versions of this question have also been posed, even just for DNFs with a layer of parity gates at the bottom  \cite{cohen2016complexity, ezrarothblum}.
Unfortunately, the best known lower bounds against \ACmodtwo circuits computing Inner Product are quite weak.
The state of the art result \cite{CGJWX} for any constant depth $d>4$  is that Inner Product cannot be computed by
any depth-$(d+1)$ \ACmodtwo circuit of size $O(n^{1+4^{-(d+1)}})$. 
We show that \Cref{thm:approxdegLC} implies an improved (if still unsatisfying) lower bound of $\tOmega(n^{1/(1- 2^{-d})}) = n^{1+2^{-d} + \Omega(1)}$.
More significantly, unlike prior work our lower bound holds even against circuits that compute the Inner Product function on slightly more than half of all inputs.
Below, when we refer to the depth of an \ACmodtwo\ circuit, we count the layer of parity gates toward the depth. For example,
we consider a DNF of parities to have depth 3. 

\begin{restatable}{theorem}{apptheorem}
\label{apptheorem} For any constant integer $d \geq 4$, 
any depth-$(d+1)$ \ACmodtwo circuit computing the
 Inner Product function on $n$ bits on greater than a $1/2 + n^{-\log n}$ fraction of inputs has size 
$\tOmega\bigl(n^{1/(1- 2^{-d})}\bigr) 
 = n^{1+ 2^{-d} + \Omega(1)}$.
 \end{restatable}
 
This application is new and does not appear in the conference version of this paper~\cite{BKT19}.
The idea of our proof is to use the approximate degree upper bound for $\LCd$ circuits of \Cref{thm:approxdegLC} 
to show that any small \ACmodtwo circuit has non-trivial (i.e., $\gg 2^{-n}$) correlation under the uniform distribution with some parity function. Yet 
it is well-known that the Inner Product function has correlation at most $2^{-n}$ with any parity function. As we show, this rules
out the possibility that a small \ACmodtwo circuit computes the Inner Product function, even on slightly more than half
of all inputs.

\subsection{Discussion and future directions}
\label{s:discussion}
Summarizing our results, we established shared-input composition theorems for quantum
query complexity (\thm{quantum}) and approximate degree (\thm{composition}),
roughly showing that for compositions between an arbitrary function $f$ and the function $g = \AND$,
it is always possible to leverage sharing of inputs 
to obtain algorithmic speedups.
We applied these results to obtain the first sublinear upper bounds on the quantum query complexity and approximate degree of \LCd.

\para{Generalizing our composition theorems.} 
Although considering the inner function $g=\AND$ is sufficient for our applications to \LC,
an important open question is to generalize our results to larger classes of inner functions.
The proof of our composition theorem for approximate degree
actually applies
to any inner function $g$ that can be exactly represented as a low-weight sum of $\AND$s
(for example, it applies to any strongly unbalanced function $g$, meaning that
$|g^{-1}(1)|=\poly(n)$).  
Extending this further would be a major step forward in our understanding
of how quantum query complexity and approximate degree behave under composition with shared inputs.

While our paper considers the composition scenario where the top function is arbitrary and the bottom function is $\AND$, the opposite scenario is also interesting. 
Here the top function is $\AND_m$ and the bottom functions are $f_1,\ldots,f_m$, each acting on the same set of $n$ input variables. Now the question is whether we can do better than the  upper bound obtained using results on block composition that treat all the input variables as being independent. 
More concretely, for such a function $F$, the upper bound that follows from block composition is $Q(F)=O(\sqrt{m}\max_i Q(f_i))$. 
However, this upper bound cannot be improved in general, because the Surjectivity function is an example of such a function. 
Here the bottom functions $f_i$ check if the input contains a particular range element $i$, and the upper bound obtained from this argument is $O(n)$, which matches the lower bound~\cite{beame,sherstov15}. 
Surprisingly, this lower bound only holds for quantum query complexity, as we know that the approximate degree of Surjectivity is $\tilde{\Theta}(n^{3/4})$. We do not know if the upper bound obtained from block composition can be improved for approximate degree.

\para{Quantum query complexity of \LC and DNFs.} 
For quantum query complexity, we obtain the upper bound $Q(\LCd)=\tilde{O}(n^{1-2^{-d}})$,
nearly matching the lower bound $Q(\LCd)=n^{1-2^{-\Omega(d)}}$ from \cite{CKK12}. However, the bounds do not match for any fixed value of $d$. The lack of matching lower bounds can be attributed to the fact that the Surjectivity function,
which is known to have linear quantum query complexity, is computed by a quadratic-size depth-3 circuit,
rather than a quadratic-size depth-2 circuit (i.e., a DNF). If one could prove a linear lower bound on
the quantum query complexity of some quadratic-size DNF, the argument of \cite{CKK12}
would translate this into a $\tilde{\Omega}(n^{1-2^{-d}})$ lower bound for \LCd, matching our upper bound. 
Unfortunately, no linear lower bound on the quantum query complexity of \emph{any} polynomial size DNFs is known;
we highlight this as an important open problem (the same problem was previously been posed by Troy Lee 
with different motivations \cite{troylee}).

\begin{openproblem}
Is there a polynomial-size DNF with $\tOmega(n)$ quantum query complexity?
\end{openproblem}

The quantum query complexity of depth-2 \LC, or linear-size DNFs also remains open. The best upper bound is $O(n^{3/4})$, but the best lower bound is $\Omega(n^{0.555})$~\cite{CKK12}. Any improvement in the lower bound would also imply, in a black-box way, an improved lower bound for the Boolean matrix product verification problem.
Improving the lower bound all the way to $\Omega(n^{3/4})$ would imply optimal lower bounds for all of \LC using the argument in \cite{CKK12}.
We conjecture that there is a linear-size DNF with quantum query complexity $\Omega(n^{3/4})$, matching the known upper bound.

\para{Approximate degree of \LC and DNFs.}
For approximate degree, we obtain the upper bound $\adeg(\LCd)=\tilde{O}(n^{1-2^{-d}})$,
and prove a new lower bound of $\adeg(\LCd)=n^{1-2^{-\Omega(\sqrt{d})}}$.
The reason our approximate degree lower bound approaches $n$ more slowly
than the quantum query lower bound from \cite{CKK12} is that, while the quantum query complexity of \AC
is known to be $\Omega(n)$, such a result is not known for approximate degree. 
This remains an important open problem.

\begin{openproblem}
	Is there a problem in \AC with approximate degree $\tOmega(n)$?
\end{openproblem}

Our lower bound argument would translate, in a black-box manner, any linear lower bound on the approximate degree of a general \AC
circuit into a nearly tight lower bound for \LCd. 

Alternatively, it would be very interesting if one could improve our approximate degree upper bound for \LCd.
Even seemingly small improvements to our upper bound would have significant implications.
Specifically, standard techniques (see, e.g., \cite{chaudhuri1996deterministic}) imply that for any constant $\delta>0$, there are approximate majority functions\footnote{Here, by an approximate
majority function, we mean any total function $f$ on $n$ bits for which there exist constants $0 < p < 1/2 < q$ such that $|x|\leq pn \Longrightarrow f(x)=0$ and $|x|\geq qn \Longrightarrow f(x)=1$.} 
computable by depth-$(2d+3)$ circuits of size $O(n^{1+2^{-d}+\delta})$.\footnote{This
precise result has not appeared in the literature; we prove it in \Cref{app:amaj} for completeness.}
This means that, for sufficiently large constant $d$, if one could improve our upper bound on the approximate degree of \LCd
from $\tO(n^{1-2^{-d}})$  to $\tO(n^{1-2^{-d/2.001}})$, one would obtain a sublinear upper bound on the approximate degree
of some total function computing an approximate majority. This would answer a question of Srinivasan \cite{filmus2014real}, and may be considered a
surprising result, as approximate majorities are currently
the primary natural candidate \AC functions that may exhibit linear approximate degree \cite{BKT18}.

\subsection{Paper organization and notation}

This paper is organized so as to be accessible to readers without familiarity with quantum algorithms. \sec{quantum} assumes the reader is somewhat familiar with quantum query complexity and Grover's algorithm~\cite{Gro96}, but only uses Grover's algorithm as a black box. In \sec{quantum} we show our main result on the quantum query complexity of shared-input compositions (\thm{quantum}). \sec{approx} proves our result about the approximate degree of shared-input compositions (\thm{composition}).  \sec{LC} uses the results of these sections (in a black-box manner) to upper bound the quantum query complexity and approximate degree of \LC circuits, and proves related lower bounds. \sec{agnostic} uses the results of \sec{LC} to obtain algorithms to agnostically PAC learn \LC circuits.
\Cref{s:app} derives our average-case lower bounds on the size of \ACmodtwo\ circuits computing the Inner Product function. This section is new and does not appear in the conference version of this paper~\cite{BKT19}.

In this paper we use the $\tO(\cdot)$ and $\tOmega(\cdot)$ notation to suppress logarithmic factors. More formally, $f(n)=\tO(g(n))$ means there exists a constant $k$ such that $f(n)=O(g(n)\log^k g(n))$, and similarly $f(n)=\tOmega(g(n))$ means there exists a constant $k$ such that $f(n)=\Omega(g(n)/\log^k g(n))$. For a string $x\in\B^n$, we use $|x|=\sum_i x_i$ to denote the Hamming weight of $x$, i.e., the number of entries in $x$ equal to $1$. For any positive integer $n$, we use $[n]$ to denote the set $\{1,2,\ldots,n\}$.
Given two functions $f_m, g_k$,  let $f_m \circ g_k \colon \B^{m \cdot k} \to \B$ denote their \emph{block composition}, i.e.,
$(f_m \circ g_k)(x) = f_m(g_k(x_1), \dots, g_k(x_m))$, where for every $i\in[m]$, $x_i$ is a $k$-bit string. For non-negative integers $n$ and $k$,
we use $\binom{n}{\leq k}$ to denote $\sum_{i=0}^k \binom{n}{i}$. A basic fact is that $\binom{n}{\leq k} \leq n^k$. 

\section{Quantum algorithm for composed functions}
\label{sec:quantum}

\subsection{Preliminaries}

As described in the introduction, our quantum algorithm only uses variants of Grover's algorithm~\cite{Gro96} and is otherwise classical. 
To make this section accessible to those without familiarity with quantum query complexity, we only state the minimum required preliminaries to understand the algorithm. 
Furthermore, we do not optimize the logarithmic factors in our upper bound to simplify the presentation.
For a more comprehensive introduction to quantum query complexity, we refer the reader to the survey by Buhrman and de Wolf~\cite{dtsurvey}.

In quantum or classical query complexity, the goal is to compute some known function $f:\B^n \to \B$ on some unknown input $x\in\B^n$ while reading as few bits of $x$ as possible. 
Reading a bit of $x$ is also referred to as ``querying'' a bit of $x$, and hence the goal is to minimize the number of queries made to the input.

For example, the deterministic query complexity of a function $f$ is the minimum number of queries needed by a deterministic algorithm in the worst case. A deterministic algorithm must be correct on all inputs, and can decide which bit to query next based on the input bits it has seen so far. Another example of a query model is the bounded-error randomized query model. The bounded-error randomized query complexity of a function $f$, denoted $R(f)$, is the minimum number of queries made by a randomized algorithm that computes the function correctly with probability greater than or equal to $2/3$ on each input. In contrast to a deterministic algorithm, such an algorithm has access to a source of randomness, which it may use in deciding which bits to query.

The bounded-error quantum query complexity of $f$, denoted $Q(f)$, is similar to bounded-error randomized query complexity, except that the algorithm is now quantum. In particular, this means the algorithm may query the inputs in superposition. Since quantum algorithms can also generate randomness, for all functions we have $Q(f) \leq R(f)$.

An important example of the difference between the two models is provided by the $\OR_n$ function, which asks if any of the input bits is equal to 1. We have $R(\OR_n)=\Theta(n)$, because intuitively if the algorithm only sees a small fraction of the input bits and they are all $0$, we do not know whether or not the rest of the input contains a $1$. However, Grover's algorithm is a quantum algorithm that solves this problem with only $O(\sqrt{n})$ queries~\cite{Gro96}. The algorithm is also known to be tight, and we have $Q(\OR_n)=\Theta(\sqrt{n})$~\cite{BBBV97}. 

There are several variants of Grover's algorithm that solve related problems and are sometimes more useful than the basic version of the algorithm. Most of these can be derived from the basic version of Grover's algorithm (and this sometimes adds logarithmic overhead).  

In this work we need a variant of Grover's algorithm that finds a $1$ in the input faster when there are many $1$s. Let the Hamming weight of the input $x$ be $t = |x|$. If we know $t$, then we can use Grover's algorithm on a randomly selected subset of the input of size $O(n/t)$, and one of the $1$s will be in this set with high probability. Hence the algorithm will have query complexity $O(\sqrt{n/t})$. With some careful bookkeeping, this can be done even when $t$ is unknown, and the algorithm will have expected query complexity $O(\sqrt{n/t})$. More formally, we have the following result of Boyer, Brassard, H{\o}yer, and Tapp~\cite{BBHT98}.

\begin{lemma}\label{lem:BBHT}
Given query access to a string $x\in\B^n$, there is a quantum algorithm that when $t = |x| > 0$, always outputs an index $i$ such that $x_i=1$ and makes $O(\sqrt{{n}/{t}})$ queries in expectation. When $t = 0$, the algorithm does not terminate.
\end{lemma}

Note that because we do not know $t=|x|$, we only have a guarantee on the expected query complexity of the algorithm, not the worst-case query complexity. Note also that this variant of Grover's algorithm is a zero-error algorithm in the sense that it always outputs a correct index $i$ with $x_i = 1$ when such an index exists.

In our algorithm we use an amplified version of the algorithm of \lem{BBHT}, which adds a log factor to the query complexity and always terminates after $O(\sqrt{n}\log n)$ queries.

\begin{lemma}\label{lem:ampBBHT}
	Given query access to a string $x\in\B^n$, there is a quantum algorithm that 
	\begin{enumerate}
		\item when $|x| = 0$, the algorithm always outputs ``$|x|=0$'',
		\item when $|x|>0$, it outputs an index $i$ with $x_i=1$ with probability $1-\frac{1}{\poly(n)}$, and
		\item terminates after $O\(\sqrt{\frac{n}{|x|+1}} \log n\)$ queries with probability $1-\frac{1}{\poly(n)}$.
	\end{enumerate}
\end{lemma}

\begin{proof}
This algorithm is quite straightforward. We simply run $O(\log n)$ instances of the algorithm of \lem{BBHT} in parallel and halt if any one of them halts. If we reach our budget of $O(\sqrt{n}\log n)$ queries, then we halt and output ``$|x|=0$''.

Let us argue that the algorithm has the claimed properties. First, since the algorithm of \lem{BBHT} does not terminate when $|x|=0$, our algorithm will correctly output ``$|x|=0$'' at the end for such inputs. When $|x|>0$, we know that the algorithm of \lem{BBHT} will find an index $i$ with $x_i=1$ with high probability after $O(\sqrt{n})$ queries. The probability that $O(\log n)$ copies of this algorithm do not find such an $i$ is exponentially small in $O(\log n)$, or polynomially small in $n$. Finally, our algorithm makes only $O(\sqrt{n}\log n)$ queries when $|x|=0$ by construction. When $|x|>0$, we know that the algorithm of \lem{BBHT} terminates after an expected $O(\sqrt{{n}/{|x|}})$ queries, and hence halts with high probability after $O(\sqrt{{n}/{|x|}})$ queries by Markov's inequality. The probability that none of $O(\log n)$ copies of the algorithm halt after making $O(\sqrt{{n}/{|x|}})$ queries each is inverse polynomially small in $n$ again.
\end{proof}

\subsection{Quantum algorithm}

We are now ready to present our main result for quantum query complexity, which we restate below.

\quantum*

While \thm{quantum} allows the bottom $\AND$ gates to depend on negated variables, it will be without loss of generality in the proof to assume that all input variables are unnegated. This is because we can instead work with the function $h' : \B^{2n} \to \B$ obtained by treating the positive and negative versions of a variable separately, increasing our final quantum query upper bound by a constant factor.

We now define some notation that will aid with the description and analysis of the algorithm. 
We know that our circuit $h$ has $m$ $\AND$ gates and $n$ input bits $x_i$. We say an $\AND$ gate has \emph{high fan-in} if the number of inputs to that $\AND$ gate is greater than or equal to $n/Q(f)$. 
Note that if our circuit $h$ has no high fan-in gates, then we are done, because we can simply use the upper bound for block composition, i.e., $Q(f\circ g) = O(Q(f)Q(g))$, to compute $h$, since we will have $Q(h) = O(Q(f)\times \sqrt{n/Q(f)}) = O(\sqrt{Q(f)\cdot n})$.

Our goal is to reduce to this simple case. More precisely, we will start with the given circuit $h$, make some queries to the input, and then simplify the given circuit to obtain a new circuit $h'$. The new circuit will have no high fan-in gates, but will still have $h'(x)=h(x)$ on the given input $x$. Note that $h'$ and $h$ have the same output only for the given input $x$, and not necessarily for all inputs.

For any such circuit $h$, let $S \subseteq [m]$ be the set of all high fan-in $\AND$ gates, and let $w(S)$ be the total fan-in of $S$, which is the sum of fan-ins of all gates in $S$. In other words, it is the total number of wires incident to the set $S$. Since the set $S$ only has gates with fan-in at least $n/Q(f)$, we have 
\begin{equation}
	w(S) \geq n|S|/Q(f).
\end{equation}

We now present our first algorithm, which is a subroutine in our final algorithm. This algorithm's goal is to take a circuit $h$, with $|S|$ high fan-in gates and $w(S)$ wires incident on $S$, and reduce the size of $w(S)$ by a factor of $2$. Ultimately we want to have $|S|=w(S)=0$, and hence if we can decrease the size of $w(S)$ by $2$, we can repeat this procedure logarithmically many times to get $|S|=w(S)=0$.

\begin{lemma}
	\label{lem:halving}
	Let $h:\B^n \to \B$ be a depth-2 circuit where the top gate is a function $f:\B^m \to \B$ and the bottom level gates are  $\AND$ gates on a subset of the input bits and their negations (as depicted in \fig{composition}). Let $w(S)$ be the total fan-in of all  high fan-in gates in $h$ (i.e., gates with fan-in $\geq n/Q(f)$).
	
	Then there is a quantum query algorithm that makes $O(\sqrt{Q(f)\cdot n} \log n)$ queries to $x\in\B^n$ and outputs a new circuit $h'$ of the same form such that $w(S')\leq w(S)/2$, where $w(S')$ is the total fan-in of all high fan-in gates in $h'$, and such that with probability $1-\frac{1}{\poly(n)}$ (over the internal randomness of the algorithm) we have $h(x)=h'(x)$ for the query input $x$.\footnote{The new circuit $h'$ is only promised to satisfy $h(x)=h'(x)$ on the specific query input $x$ on which this algorithm is run.} 
\end{lemma}

\begin{proof}
	The overall structure of the claimed algorithm is the following: We query some well-chosen input bits, and on learning the values of these bits, we simplify the circuit accordingly. If an input bit is 0, then we delete all the $\AND$ gates that use that input bit.
	If an input bit is 1, we delete all outgoing wires from that input bit since a 1-input does not affect the output of an $\AND$ gate.
	
	Since the circuit will change during the algorithm, let us define $S_0$ to be the initial set of high fan-in (i.e., gates with fan-in $\geq n/Q(f)$) $\AND$ gates in $h$. 

	We also define the degree of an input $x_i$, denoted $\deg(i)$, to be the number of high fan-in $\AND$ gates that it is an input to. Note that this is not the total number of outgoing wires from $x_i$, but only those that go to high fan-in $\AND$ gates, i.e., gates in the set $S$. With this definition, note that $\sum_{i\in[n]} \deg(i) = w(S)$, for any circuit. 
	We say an input bit $x_i$ is \emph{high degree} if $\deg(i)\geq |S_0|/(2Q(f))$. This value is chosen since it is at least half the average degree of all $x_i$ in the initial circuit $h$. 
	As the algorithm progresses, the circuit will change, and some inputs that were initially high degree may become low degree as the algorithm progresses, but a low degree input will never become high degree. 
	But note that the definition of a high-degree input bit does not change, since it only depends on $S_0$ and $Q(f)$, which are fixed for the duration of the algorithm. 
	
	Finally, we call an input bit $x_i$ is \emph{marked} if $x_i=0$. We are now ready to describe our algorithm	by the following pseudocode (see \alg{basic}). 
	\begin{algorithm}
		\caption{The algorithm of \lem{halving}.\label{alg:basic}}
		\begin{algorithmic}[1]
			\Statex
			\Let{$S_0$}{Set of high fan-in $\AND$ gates in $h$} 
			\Repeat 
			\Let{$M$}{Set of high-degree marked inputs \algorithmiccomment{$M := \left\{i:   x_i =0 \wedge \deg(i)\geq \frac{|S_0|}{(2Q(f))}\right\}$}}
			\State{Grover Search for an index $i$ in $M$}
			\If{we find such an $i$}
			\State{Delete all $\AND$ gates that use $x_i$ as an input}
			\EndIf
			\Until{Grover Search fails to find an $i \in M$}
			\State{Delete all remaining high-degree inputs and all outgoing wires from these inputs}
		\end{algorithmic}
	\end{algorithm}
	
	In more detail, we repeatedly use the version of Grover's algorithm in \lem{ampBBHT} to find a high-degree marked input, which is an input $x_i$ such that $x_i=0$ and $\deg(i)\geq \frac{|S_0|}{2Q(f)}$. If we find such an input, we delete all the $\AND$ gates that use $x_i$ as an input, and repeat this procedure. Note that when we repeat this procedure, the circuit has changed, and hence the set of high-degree input bits may become smaller. The algorithm halts when Grover's algorithm is unable to find any high-degree marked inputs. At this point, all the high-degree inputs are necessarily unmarked with very high probability, which means they are set to $1$. We can now delete all these input bits and their outgoing wires because $\AND$ gates are unaffected by input bits set to $1$.
	
	Let us now argue that this algorithm is correct.  Let $S'$ denote the set of high fan-in $\AND$ gates in the new circuit $h'$ obtained at the end of the algorithm, and $w(S')$ be the total fan-in of gates in $S'$. 
	Note that when the algorithm  terminates, there are no high-degree inputs (marked or unmarked). Hence every input bit that has not been deleted has $\deg(i)<\frac{|S_0|}{2Q(f)}$. Since there are at most $n$ input bits, we have 
	\begin{equation}
		w(S')=\sum_{i \in [n]} \deg(i) < \frac{n}{2Q(f)}|S_0|.
	\end{equation}  
	But we also know that we started with $w(S)\geq n|S_0|/Q(f)$, since each gate in $S_0$ has fan-in at least $n/Q(f)$. Hence $w(S')\leq w(S)/2$, which proves that the algorithm is correct.
	
	We now analyze the query complexity of this algorithm. Let the loop in the algorithm execute $r$ times. It is easy to see that $r \leq 2Q(f)$ because each time a high-degree marked input is found, we delete all the $\AND$ gates that use it as an input, which is at least $|S_0|/(2Q(f))$ gates. Since there were at most $S_0$ gates to begin with, this procedure can only repeat $2Q(f)$ times.
	
	When we run Grover's algorithm to search for a high-degree marked input bit $x_i$ in the first iteration of the loop, suppose there are $k_1$ high-degree marked inputs. Then the variant of Grover's algorithm in \lem{ampBBHT} finds a marked high-degree input and makes $O(\sqrt{n/k_1}\log n)$ queries with probability $1-\frac{1}{\poly(n)}$. In the second iteration of the loop, the number of high-degree marked inputs, $k_2$, has decreased by at least one. It can also decrease by more than 1 since we deleted several $\AND$ gates, and some high-degree inputs can become low-degree. In this iteration, our variant of Grover's algorithm (\lem{ampBBHT}) makes $O(\sqrt{n/k_2} \log n)$ queries, and we know that $k_1>k_2$. This process repeats and we have $k_1>k_2>\cdots >k_r$. Since there was at least one high-degree marked input in the last iteration, $k_r\geq 1$. Combining these facts we have for all $j\in[r]$, $k_j \geq r-j+1$. Thus the total expected query complexity is
	\begin{equation}
		O\(\sum_{j=1}^r \sqrt{\frac{n}{k_j}} \log n\) 
		= O\(\sum_{j=1}^r \sqrt{\frac{n}{r-j+1}} \log n\) 
		= O\(\sqrt{n} \sum_{j=1}^r \frac{1}{\sqrt{j}} \log n\) 
		= O\(\sqrt{nr} \log n\),
	\end{equation}
	which is $O\(\sqrt{n \cdot Q(f)} \log n\).$
	We now have a quantum query algorithm  that satisfies the conditions of the lemma with probability at least $1-\frac{1}{\poly(n)}$.
\end{proof}

We are now ready to prove \thm{quantum}.

\begin{proof}[Proof of \protect{\thm{quantum}}]
We start by applying the algorithm in \lem{halving} to our circuit as many times as needed to ensure that set $S$ is empty. Since each run of the algorithm reduces $w(S)$ by a factor of 2, and $w(S)$ can start off being as large as $m\cdot n$, where $m$ is the number of $\AND$ gates and $n$ is the number of inputs, we need to run the algorithm $\log(mn)$ times. Since the algorithm of \lem{halving} is correct with probability $1-\frac{1}{\poly(n)}$, we do not need to boost the success probability of the algorithm. The total number of queries needed to ensure $S$ is empty is $O(\sqrt{Q(f)\cdot n}\log(n)\log(mn))$.

Now we are left with a circuit $h'$ with no high fan-in $\AND$ gates. That is, all $\AND$ gates have fan-in at most $n/Q(f)$. We now evaluate $h'$ using the standard composition theorem for disjoint sets of inputs, which has query complexity
\begin{equation}
	O(Q(f) \cdot Q(\AND_{n/Q(f)}))= O(Q(f) \cdot \sqrt{n/Q(f)}) = O\(\sqrt{Q(f)\cdot n}\).
\end{equation}
The total query complexity is 
$O(\sqrt{Q(f)\cdot n}\log(n)\log(mn)) = O(\sqrt{Q(f)\cdot n}\log^2(mn))$.
\end{proof}

Note that we have not attempted to reduce the logarithmic factors in this upper bound. We believe it is possible to make the quantum upper bound match the upper bound for approximate degree with a more careful analysis and slightly different choice of parameters in the algorithm.

\section{Approximating polynomials for composed functions}
\label{sec:approx}

\subsection{Preliminaries} 
We now define the various measures of Boolean functions and polynomials that we require in this section. Since we only care about polynomials approximating Boolean functions, we focus without loss of generality on multilinear polynomials as any polynomial over the domain $\B^n$ can be converted into a multilinear polynomial (since it never helps to raise a Boolean variable to a power greater than $1$).

The approximate degree of a Boolean function, commonly denoted $\adeg(f)$, is the minimum degree of a polynomial that entrywise approximates the Boolean function. It is a basic complexity measure and is known to be polynomially related to a host of other complexity measures such as decision tree complexity, certificate complexity, and quantum query complexity~\cite{dtsurvey,bt21}. We also use another complexity measure of polynomials, which is the sum of absolute values of all the coefficients of the polynomial. This is the query analogue of the so-called $\mu$-norm used in communication complexity~\cite[Definition 2.7]{LS09}. We now formally define these measures.

\begin{definition}
	Let $p:\R^n \to \R$ be a multilinear polynomial 
\begin{equation}
p(x_1,\ldots,x_n)=\sum_{s\in\Bo^n} \alpha_s x_1^{s_1}\cdots x_n^{s_n}.
\end{equation} 
	We define the following complexity measures of the polynomial $p$: 
	\begin{align}
	\deg(p) = \max \Bigl\{\sum_{i\in[n]} |s_i|:\alpha_s\neq 0\Bigr\} \qquad \mathrm{and} \qquad
	\mon(p) = \sum_{s\in\Bo^n} |\alpha_s|.
	\end{align}
	For a Boolean function $f:\B^n \to \B$, we define the following complexity measures:
	\begin{align}
	\deg_{\eps}(f) &= \min \{\deg(p):\forall x\in\B^n,~|f(x)-p(x)|\leq \eps\}\\
	\mu_{\eps}(f) &= \min \{\mon(p):\forall x\in\B^n,~|f(x)-p(x)|\leq \eps\}
	\end{align}
Finally, we define $\adeg(f)=\deg_{1/3}(f)$ and $\amu(f)=\mu_{1/3}(f)$.
\end{definition}

We use the following standard relationship between the two measures in our results.

\begin{lemma}
	\label{lem:degreebound}
	For any multilinear polynomial $p:\R^n \to \R$ such that $|p(x)|=O(1)$ for all $x\in\B^n$, we have
\begin{equation}
\log \mu(p) = O(\deg(p)\log n).\label{eq:mudeg}
\end{equation}
	Consequently, for any Boolean function $f:\B^n \to \B$ and $\eps\in[0,1/3]$, we have 
\begin{equation}
\log \mu_\eps(f) = O(\deg_\eps(f)\log n).
\end{equation}
\end{lemma}

\begin{proof}
	First let us switch to the $\Bf$ representation instead of the $\Bo$ representation we have used so far. Let $y_i = (-1)^{x_i}$, and replace every occurrence of $x_i$ in the polynomial $p$ with $\frac{1}{2}(1+y_i)$ to obtain a multilinear polynomial $p(y_1,\ldots,y_n)=\sum_{s\in\Bo^n} \beta_s y_1^{s_1}\cdots y_n^{s_n}$. In this representation, a coefficient $\beta_s$ is simply the expectation over the hypercube of the product of $p$ and a parity function, and hence is at most $O(1)$ in magnitude. Since there are only $\binom{n}{\leq \deg(p)} \leq n^{\deg(p)}$ monomials, the sum of absolute values of all coefficients is $O(n^{\deg(p)})$.
	
	When we switch from this representation back to the $\Bo$ representation, we replace every $y_i$ with $2x_i-1$. Consider this transformation on a single monomial with coefficient $1$. This converts the monomial of degree $d$ into a polynomial over those $d$ variables, such that the sum of coefficients in this polynomial is at most $3^d$. Thus the sum of absolute values of all coefficients is $\mu(p)=O(3^{\deg(p)} n^{\deg(p)}) = n^{O(\deg(p))}$, which proves \eq{mudeg}.

	Now consider any Boolean function $f:\B^n\to\B$, and a multilinear polynomial $p$ that minimizes $\deg_\eps(f)$. We can apply \eq{mudeg} to this polynomial to obtain $\log \mu(p) = O(\deg(p)\log n)$. Since $\deg(p) = \deg_\eps(f)$ by assumption, and $\mu_\eps(f)\leq\mu(p)$, since $\mu_\eps(f)$ minimizes over all $\eps$-approximating polynomials, we get $\log \mu_\eps(f) = O(\deg_\eps(f)\log n)$.
\end{proof}

This shows that $\log \mu(p)$ is at most $\deg(p)$ (up to log factors). However, $\log \mu(p)$ may be much smaller than $\deg(p)$, as evidenced by the polynomial $p(x)=x_1\cdots x_n$. Similarly, $\log \amu(f)$ may be much smaller than $\adeg(f)$, as evidenced by the $\AND$ function on $n$ bits, which has $\adeg(\AND_n)=\Theta(\sqrt{n})$ \cite{nisanszegedy}, but $\amu(\AND_n)\leq 1$.

\subsection{Polynomial upper bound}

In this section we prove \thm{composition}, which
follows from the following more general composition theorem.

\begin{restatable}{theorem}{composition-full}
\label{thm:composition-full}
Let $h:\B^n\to \B$ be computed by a depth-2 circuit where the top gate is a function $f:\B^m \to \B$ and the bottom level gates are  $\AND$ gates on a subset of the input bits and their negations (as depicted in \fig{composition}). Then 
\begin{equation}
\deg_\eps(h) = O\(\sqrt{n\log{\mu_\eps(f)}}+\sqrt{n\log(1/\eps)}\)=O\(\sqrt{n\deg_\eps(f)\log m}+\sqrt{n\log(1/\eps)}\).\label{eq:comp}
\end{equation}
\end{restatable}

\begin{proof}
	Let us first fix some notation. We will use $x\in\B^n$ to refer to the input of the full circuit $h:\B^n \to \B$. Let the inputs to the top $f:\B^m \to \B$ gate be called $y_1,\ldots, y_m$.
	
	Let $p:\B^m \to \B$ be a polynomial that minimizes $\mu_\eps(f)$. Thus we have for all $y\in\B^m$, $|p(y)-f(y)|\leq \eps$. More explicitly, $p(y_1,\ldots,y_m) = \sum_{s\in\Bo^m} \alpha_s y_1^{s_1}\cdots y_n^{s_n}$, where $\mu_\eps(f)=\sum_{s\in\Bo^m} |\alpha_s|$, and each $y_i$ is the $\AND$ of some subset of bits in $x$. Since the product of $\AND$s of variables is just an $\AND$ of all the variables involved in the product, for each $s \in\Bo^m$, there is a subset $T_s \subseteq [n]$ such that $y_1^{s_1}\cdots y_n^{s_n} = \bigwedge_{i\in T_s} x_i$.
	
	Using this we can replace all the $y$ variables in the polynomial $p$, to obtain
	\begin{equation}
	q(x)=\sum_{s\in\B^m} \alpha_s \bigwedge_{i\in T_s} x_i.
	\end{equation}
Since $p$ was an $\eps$ approximation to $f$, $q$ is an $\eps$ approximation to $h$. Now we can replace every occurrence of  $\bigwedge_{i\in T_s} x_i$ with a low error approximating polynomial for the $\AND$ of the bits in $T_s$. We know that the approximate degree of the $\AND$ function to error $\delta$ is $O(\sqrt{n\log(1/\delta)})$~\cite{smallerrorquantum}. If we approximate each $\AND$ to error $\delta=\eps/\mu_\eps(f)$, then by the triangle inequality the total error incurred by this approximation is at most $\sum_{s\in\Bo^m} |\alpha_s|\eps/\mu_\eps(f) = \eps$. 
Choosing $\delta=\eps/\mu_\eps(f)$, each $\AND$ is approximated by a polynomial of degree $O(\sqrt{n\log(1/\delta)}) = O\left(\sqrt{n\log \mu_\eps(f)}+\sqrt{n\log(1/\eps)}\right)$. Hence the resulting polynomial $q(x)$ has this degree and approximates the function $h$ to error $2\eps$. By standard error reduction techniques~\cite{BuhrmanNRW07}, we can make this error smaller than $\eps$ at a constant factor increase in the degree. This establishes the first equality in \eq{comp}, and the second equality follows from \lem{degreebound}.
\end{proof}

\section{Applications to linear-size \texorpdfstring{\AC}{AC0} circuits}
\label{sec:LC}

\subsection{Preliminaries}
A Boolean circuit is defined via a directed acyclic graph.
Vertices of fan-in 0 represent input bits, vertices of fan-out 0 represent outputs, and 
all other vertices represent one of the following logical operations: a $\mathsf{NOT}$ operation (of fan-in 1), or an unbounded fan-in $\mathsf{AND}$ or $\mathsf{OR}$ operation. 
The size of the circuit is the total number of $\mathsf{AND}$ and $\mathsf{OR}$ gates. The depth of the circuit
is the length of the longest path from an input bit to an output bit.

For any constant integer $d > 0$, \ACd refers to the class of all such circuits of polynomial size and depth $d$. \AC refers
to $\cup_{d=1}^{\infty}$\ACd. Similarly, \LCd refers to the class of all such circuits of size $O(n)$ and depth $d$,
while \LC refers to $\cup_{d=1}^{\infty}$\LCd. We will associate any circuit $C$ with the function it computes, so for example
$\adeg(C)$ denotes the approximate degree of the function computed by $C$.

It will be convenient to assume that any \ACd circuit is layered, in the sense that it consists of $d$ levels of gates which alternate between being comprised of all $\AND$ gates or all $\OR$ gates, and all negations appear at the input level of the circuit. Any \ACd circuit of size $s$ can be converted into a layered circuit of size $O(d \cdot s)$, and hence making this assumption does not change any of our upper bounds.

\subsection{Quantum query complexity} 

Applying our composition theorem for quantum algorithms (\thm{quantum}) inductively, we obtain a sublinear upper bound on the quantum query complexity of \LCd circuits.

\quantumLC*

\begin{proof}
We  prove this for depth-$d$ \LC circuits by induction on $d$. The base case is $d=1$, where the function is either $\AND$ or $\OR$ on $n$ variables, both of which have quantum query complexity $O(\sqrt{n})$~\cite{Gro96}.

Now consider a function $h$, which is a layered depth-$d$ \AC circuit of size $O(n)$. It can be written as a depth-$2$ circuit (as in \thm{quantum}) where the top function is a \LC circuit $f$ of depth $d-1$ on at most $O(n)$ inputs, and the bottom layer has only $\AND$ gates. (If the bottom layer has $\OR$ gates we can consider the negation of the function without loss of generality, since the quantum query complexity of a function and its negation is the same.)

By the induction hypothesis we know that the quantum query complexity of any depth-$(d-1)$, size-$O(n)$ \AC circuit with $O(n)$ inputs is $\tO(n^{1-2^{-(d-1)}})$. 
Invoking \thm{quantum}, we have that the quantum query complexity of the depth-$d$ function $h$ is $\tO\bigl(n^{1-2^{-d}}\bigr)$.
\end{proof}

\subsection{Approximate degree upper bound}

We can now prove \Cref{thm:approxdegLC}, restated below for convenience:
\approxdegLC*

This follows from a more general result:

\begin{restatable}{theorem}{upperbound}
\label{thm:upper}
For any function $h:\B^n\to\B$ computed by an  \AC circuit of size $s\geq 1$ and depth $d\geq 1$, 
we have 
\begin{equation}
\deg_\eps(h)=\begin{cases} O\(\sqrt{n\log(1/\eps)}\) &\mbox{if } \eps \leq 2^{-s} \Leftrightarrow \log(1/\eps)\geq s\\
\tO\(\sqrt{n} s^{{1}/{2}-2^{-d}} \({\log(1/\eps)}\)^{2^{-d}}\) & \mbox{if } \eps > 2^{-s} \Leftrightarrow \log(1/\eps) < s \end{cases}.
\end{equation}
    In particular, for any $h\in\LCd$, we have  $\adeg(h)=\tO\bigl(n^{1-2^{-d}}\bigr)$.
\end{restatable}

\begin{proof}
We  prove this for depth-$d$ \AC circuits by induction on $d$. The base case is $d=1$, where the function is either $\AND$ or $\OR$ on $n$ variables, both of which have $\eps$-approximate degree $O(\sqrt{n\log (1/\eps)})$~\cite{smallerrorquantum}.

Now consider a function $h$, which is a general depth-$d$ \AC circuit of size $s$. It can be written as a depth-$2$ circuit (as in \thm{composition}) where the top function is a size-$s$ \AC circuit $f$ of depth $d-1$ on at most $s$ inputs, and the bottom layer has only $\AND$ gates. If the bottom layer has $\OR$ gates we can consider the negation of the function without loss of generality, since the $\eps$-approximate degree of a function and its negation is the same.

In the first case, if $\eps \leq 2^{-s}$, then for any function $f:\B^s\to\B$  there is a polynomial of degree $s$ and sum of coefficients at most $2^s$ that exactly equals $f$ on all Boolean inputs. Hence we can apply \thm{composition} to get that $\deg_\eps(h)=O(\sqrt{ns}+\sqrt{n\log(1/\eps)}) = O(\sqrt{n\log(1/\eps)})$.

In the second case, if $\eps > 2^{-s}$, by the induction hypothesis we know that the $\eps$-approximate degree of any depth-$(d-1)$, size-$O(s)$ \AC circuit with $s$ inputs is $\tO(s^{1-2^{-(d-1)}}(\log(1/\eps))^{2^{-(d-1)}})$. 
Invoking \thm{composition}, we have that the approximate degree of the depth-$d$ function is 
\begin{equation}
\tO\(\sqrt{n s^{1-2^{-(d-1)}}(\log(1/\eps))^{2^{-(d-1)}}}+\sqrt{n\log(1/\eps)}\)=\tO\(\sqrt{n} s^{1/2-2^{-d}}(\log(1/\eps))^{2^{-d}}\).\qedhere
\end{equation}
\end{proof}

\subsection{Approximate degree lower bound}
\label{sec:lower}

In this section we prove our lower bound on the approximate degree of \LCd, restated below for convenience.

\lowerbound*

Before proving the theorem, we will need to introduce several lemmas. The first lemma follows from the techniques of \cite{AB84} (see \cite{Kop13} for an exposition).

\begin{lemma}\label{lem:gapmaj}
There exists a Boolean circuit $C$ with $n$ inputs, of depth 3, and size $\tO(n^2)$ satisfying the following two properties:
\begin{itemize}
\item $C(x) = 0$ for all $x$ of Hamming weight at most $n/3$.
\item $C(x) = 1$ for all $x$ of Hamming weight at least $2n/3$.
\end{itemize}
\end{lemma}
We refer to the function computed by the circuit $C$ of \lem{gapmaj} as $\mathsf{GAPMAJ}$, short for a gapped majority function (such a function is sometimes also called an \emph{approximate majority} function).

The following lemma of \cite{bchtv} says that if $f$ has large $\eps$-approximate degree for $\eps=1/3$, then block-composing $f$ with $\mathsf{GAPMAJ}$ on $O(\log n)$ bits
yields a function with just as high $\eps'$-approximate degree, with $\eps'$ very close to $1/2$. 

\begin{lemma}[\cite{bchtv}] \label{lemmaa}
Let $f \colon \B^n \to \B$ be any function. Then for $\eps=1/2-1/n^2$, $\deg_{\eps}(\mathsf{GAPMAJ}_{10 \log n} \circ f) \geq \adeg(f)$.
\end{lemma}

The following lemma says that if $f$ has large $\eps$-approximate degree for $\eps$ very close to $1/2$, then block-composing any function $g$ with $f$ results
in a function of substantially
larger approximate degree than $g$ itself. 
\begin{lemma}[\cite{sherstovhalfspaces1}] \label{lemmab}
Let $g \colon \B^m \to \B$ and $f \colon \B^n \to \B$ be any functions. Then $\adeg(g \circ f) \geq \adeg(g) \cdot \deg_{1/2-1/m^2}(f)$.
\end{lemma}

Combining Lemmas \ref{lemmaa} and \ref{lemmab}, we conclude:
\begin{corollary} \label{nomorecors}
Let $g \colon \B^m \to \B$ and $f \colon \B^n \to \B$ be any functions. Then $\adeg(g \circ \mathsf{GAPMAJ}_{10\log n} \circ f) \geq \adeg(g) \cdot \adeg(f)$.
\end{corollary}

We are now ready to prove \Cref{thm:lower}, which is restated at the beginning of this section.

\begin{proof}[Proof of \Cref{thm:lower}]
	Let $\ell \geq 1$ be any constant integer to be specified later (ultimately, we will
	set $\ell=\Theta(\sqrt{d})$, where $d$ is as in the statement of the theorem). \cite{BKT18} exhibit a circuit family $C^* \colon \B^n \to \B$ of depth at most $3\ell$, size at most $n^2$, and approximate degree satisfying
	$\adeg(C^*) \geq D$ for some $D\geq \tOmega(n^{1-2^{-\ell}})$. 
	We need to transform this quadratic-size circuit into a circuit $C$ of \emph{linear} size, without substantially reducing its approximate degree, or substantially increasing its depth (in
	particular, the depth of $C$ should be at most $d$).

	To accomplish this, we apply the following iterative transformation. 
	At each iteration $i$, we produce a new circuit $C^i \colon \B^n \to \B$ of linear size, such that $\adeg(C^i)$ gets closer and closer to $\adeg(C)$ as $i$ grows. 
	Our final circuit will be $C:=C^{\ell}$.

	$C^1$ is defined to simply be $\OR_n$, which is clearly in \LCone. 

	The transformation from $C^{i-1}$ into $C^{i}$ works as follows. $C^i$ feeds $\sqrt{n}$ copies of $C^{i-1}_{\sqrt{n}/(10 \log n)}$ into the circuit $C^*_{\sqrt{n}} \circ \textsf{GAPMAJ}_{10\log n}$. Here, $C^{i-1}_{k}$
	denotes the function $C^{i-1}$ constructed in the previous iteration, and defined on $k$ inputs; similarly,  $C^*_{k} \colon \B^{k} \to \B^n$ refers to the function
	$C^*$ constructed by \cite{BKT18}, defined on $k$ inputs. That is:

	\begin{equation} \label{cidef}
		C^i = C^*_{\sqrt{n}} \circ \textsf{GAPMAJ}_{10\log n} \circ C^{i-1}_{\sqrt{n}/(10 \log n)}.
	\end{equation}

	Observe that $C^i$ is a function on $\sqrt{n} \cdot 10 \log n \cdot (\sqrt{n}/(10\log n)) = n$ bits.
	We now establish the following two lemmas about $C^i$.

	\begin{lemma} \label{lem:circuitlemma}
		$C^i$ is computed by a circuit of depth at most $(3\ell+3) \cdot i$, and size at most $2\cdot i \cdot n$. 
	\end{lemma}
	\begin{proof}
		Clearly this is true for $i=1$, since $C^1$ is computed by a circuit of size and depth 1.
		Assume by induction that it is true for $i-1$. Recalling that  $\textsf{GAPMAJ}_{10\log n}$ is computed by a circuit
		of size $O(\log^2 n)$ and depth 3, and $C^*_{\sqrt{n}}$ is computed by a circuit of size $n$ and depth $3\ell$,
		it is immediate from Equation \eqref{cidef} that $C^i$ is computed by a circuit 
		satisfying the following properties:
		\begin{itemize}
		\item The depth is at most $3 \ell + 3 + (3 \ell + 3)(i-1) = (3\ell + 3)i$.
		\item The size is at most $n + O(\sqrt{n} \cdot \log^2 n) + \left(\sqrt{n} \cdot 10 \log n\right) \cdot \left(2 \cdot (i-1) \cdot \sqrt{n}/(10 \log n)\right)$.
		For large enough $n$, this is at most  
		$2n + 2 \cdot (i-1) \cdot n = 2 \cdot i \cdot n$.
		\end{itemize}
	\end{proof}

	\begin{lemma} \label{lem:adeglemma}
		For $i > 1$,
		$\adeg(C^i) \geq \Omega\left( \adeg(C^*_{\sqrt{n}}) \cdot \adeg(C^{i-1}_{\sqrt{n}/(10 \log n)})\right).$
	\end{lemma}
	\begin{proof}
		Immediate from \Cref{nomorecors}.
	\end{proof}

	Since $\adeg(C^1) = \Omega(\sqrt{n})$,
	repeated application of \lem{adeglemma} implies that $\adeg(C^2) = \Omega(\sqrt{D} \cdot n^{1/4})$,
	$\adeg(C^3) = \Omega\left(\sqrt{D} \cdot (\sqrt{D} \cdot n^{1/4})^{1/2}\right) = \Omega(D^{3/4} \cdot n^{1/8})$,
	and in general, $\adeg(C^i) = \Omega\left(D^{1-2^{-i}} \cdot n^{2^{-i}}\right)$.

	Setting $i=\ell$, we obtain a circuit $C^{\ell} \colon \B^n \to \B$ with the following properties:

	\begin{itemize}
		\item By \lem{circuitlemma}, $C^{\ell}$ has size at most $2 \ell n$ and depth at most $d:=2 \ell^2$.
		\item There is a constant $c_0$ such that $C^{\ell}$ has approximate degree at least  $\Omega\left(c_0^{\ell} \cdot D^{1-2^{-\ell+1}} \cdot n^{2^{-\ell}}\right) \geq \Omega(c_0^{\ell} \cdot n^{1-2^{-\ell+1/2}})$.
	\end{itemize}

	Hence, for any constant value of $d=2\ell^2$, we have constructed a circuit of depth $d$, size $O(n)$, and approximate degree at least $\Omega(n^{1-2^{-\Omega(\sqrt{d})}})$,
	as required by the theorem.
\end{proof}

\subsection{Sublinear-size circuits of arbitrary depth}
\label{sec:sublinear}

\thm{quantum} and \Cref{thm:composition} also allow us to prove sublinear quantum query complexity and approximate degree upper bounds for arbitrary circuits of sublinear size. 

\begin{theorem} \label{thm:sublinearadeg}
Let $h : \B^n \to \B$ be computed by a layered circuit of size $s \le n$. Then $h$ has quantum query complexity $Q(h) = \tilde{O}(\sqrt{n s})$ and approximate degree $\adeg(h) = O(\sqrt{ns})$.
\end{theorem}

\begin{proof}
Without loss of generality, a function $h$ computed by a layered circuit of size $s\leq n$ can be written as a depth-2 circuit with a function $f:\B^s \to \B$ as the top gate and $\AND$ gates at the bottom. (The case where the bottom level consists of $\OR$ gates can be handled by negating the function.) The quantum query upper bound then follows immediately from \thm{quantum}, as $Q(f) \le s$. Moreover, for any function $f$, we have $\log \mu_0(f) = O(s)$, since the trivial polynomial obtained by adding all conjunctions over yes-inputs of $f$ satisfies this. Hence from \thm{composition} we have $\adeg(h) = O(\sqrt{ns})$.
\end{proof}

\section{Applications to agnostic PAC learning}
\label{sec:agnostic}

Our new upper bounds on the approximate degree of \LC circuits yield new subexponential time learning algorithms in the agnostic model. In this section, we provide background for, and the proof of, our main learning result restated below.

\learningresult*

\paragraph{PAC and agnostic learning models.}

In the classic Probably Approximately Correct (PAC) learning model of Valiant~\cite{Valiant:1984}, we have access
to an unknown function $f:\B^n \to \B$ from a known class of functions $\C$, called the concept class, 
through samples $(x,f(x))$, where $x$ is drawn from an unknown distribution $\D$ over $\B^n$. 
The goal is to learn a hypothesis $h:\B^n \to \B$, such that with probability $1-\delta$ (over the choice of samples), 
$h(x)$ has (Boolean) loss at most $\eps$ with respect to $\D$. Here, the Boolean loss $\mathsf{err}_{\D}(h, f)$ of $h$ is defined to be
$\Pr_{x \sim \D}[h(x) \neq f(x)] \leq \eps$.

Since the learning algorithm does not know $\D$ and is required to work for all $\D$, this model is also called the distribution-independent (or distribution-free) PAC model. 
Unfortunately, in the distribution-free setting, very few concept classes are known to be PAC learnable in 
polynomial time or even subexponential time (i.e., time $2^{n^{1-\delta}}$ for some constant $\delta>0$). 

Kearns, Schapire, and Sellie~\cite{KSS:1994} then proposed the more general (and challenging) agnostic PAC learning model, which removes the assumption that examples are determined by a function at all, let alone a function in the concept class $\C$. The learner now knows nothing about how examples are labeled, but is only required to learn a hypothesis $h$ that is at most $\eps$ worse than the best possible classifier from the class $\C$.

We now describe the agnostic PAC model more formally. Let $\D$ be any distribution on $\B^n \times \B$, and let $\C$ be a concept class, i.e., a set of Boolean
functions on $\B^n$. Define the error of $h \colon \B^n \to \B$ to be $\mathsf{err}_{\D}(h) := \Pr_{(x, b) \sim \D}[h(x) \neq y]$,
and define $\mathsf{opt} := \min_{c \in \C} \mathsf{err}_{\D}(c)$.  
We say that $\C$ is agnostically learnable in time $T(n, \eps, \delta)$ if there exists an
algorithm which takes as input $n$ and $\delta$ and has access to an example oracle $\mathsf{EX}(\D)$,
and satisfies the following properties. It runs in time at most $T(n, \eps, \delta)$, and with probability at least $1-\delta$,
it outputs a hypothesis $h$ satisfying $\mathsf{err}_{\D}(h) \leq \mathsf{opt}+\eps$. We say that the learning
algorithm runs in \emph{subexponential time} if there is some constant $\eta > 0$ such that for any
constants $\eps$ and $\delta$, the running time $T(n, \eps, \delta) \leq 2^{n^{1-\eta}}$ for sufficiently large $n$. 

The agnostic model is able to capture a range of realistic scenarios that do not fit within the standard PAC model.
In many situations it is unreasonable to know exactly that $f$ belongs to some class $\C$, since $f$ may be computed
by a process outside of our control. 
For example, the labels of $f$ may be (adversarially) corrupted by noise, resulting in a function that is no longer in $\C$. 
Alternatively, $f$ may be ``well-modeled,''
but not \emph{perfectly} modeled, by some concept in $\C$.
In fact, the agnostic learning model even allows the input sample to not be described by a function $f$ at all, 
in the sense that the distribution over the sample may have both $(x,0)$ and $(x,1)$ in its support.
This is also realistic when the model being used does not capture all of the variables on which the true function depends.

\subsection{Related work}
\label{sec:detailed}
Since the agnostic PAC model generalizes the standard PAC model, it is (considerably) harder to learn a concept class in this model. Consequently, even fewer concept classes are known to be agnostically learnable, even in subexponential time. 
For example, as mentioned in \Cref{sec:introst}, the best known algorithm for agnostically learning the simple concept class of disjunctions, which are size-$1$, depth-$1$ Boolean circuits, 
runs in time\footnote{For simplicity, we suppress runtime dependence on $\eps$ and $\delta$.} $2^{\tO(\sqrt{n})}$~\cite{kkms}. In contrast, they can be learned in polynomial time in the PAC model~\cite{Valiant:1984}. Meanwhile, several hardness results are known for agnostically learning disjunctions, including NP-hardness for proper learning~\cite{KSS:1994}, and that even improper learning is as hard as PAC learning DNF~\cite{lbw}.

While it is an important and interesting problem to agnostically learn more expressive classes of circuits in subexponential time, 
relatively few results are known. The best known general result is that all de Morgan formulas (formulas over the gate set of $\AND$, $\OR$, and $\NOT$ gates) of size $s$ can be learned in time $2^{\tO(\sqrt{s})}$~\cite{kkms, Rei11}. 
In particular, linear-size formulas (i.e., $s=\Theta(n)$) can be learned in time $2^{\tilde{O}(\sqrt{n})}$, which is the same as the best known upper bound for disjunctions. 

Even in the relatively easier PAC model, only a small number of circuit classes are known to be learnable in subexponential time. For the well-studied class of polynomial-size DNFs, or depth-2 \AC circuits, we have an algorithm running in time $2^{\tO(n^{1/3})}$~\cite{klivansservediodnfs}, and we know that new techniques will be needed to improve this bound~\cite{razborovsherstov}. Little is known about larger subclasses of \AC, other than a recent paper that studied depth-3 \AC circuits with top fan-in $t$, giving a PAC learning algorithm of runtime $2^{\tilde{O}(t\sqrt{n})}$~\cite{DRG17}, which is only subexponential when $t \ll \sqrt{n}$.

Given the current state of affairs, a subexponential-time algorithm to learn all of \AC in the standard PAC model would represent significant progress. Indeed, for $d>2$,
the fastest known PAC learning algorithm for depth-$d$ \AC circuits runs in time $2^{n-\Omega(n/\log^{d-1}n)}$ \cite{nontriviallearn}, which is quite close to the trivial runtime of $2^n$.

We view our new results for learning \LC and sublinear-size \AC circuits as intermediate steps toward this goal. We clarify that our results are incomparable to the known results about agnostically learning de Morgan formulas. A simple counting argument \cite{nisanonline} shows that there 
are linear-size DNFs that are not computable by formulas of size $o(n^2/\log n)$, so one cannot learn even depth-2 \LC
in subexponential time via the learning algorithm for de Morgan formulas. On the other hand, there are linear-size de Morgan formulas (of superconstant depth) that are not in \LC, or even \AC.

 \label{s:otherprior}

Motivated by the lack of positive results in the distribution-free PAC learning model, \cite{nontriviallearn} 
study algorithms for learning various circuit classes, with the goal of ``only'' achieving a \emph{non-trivial savings} over trivial $2^n$-time algorithms.
By achieving non-trivial savings, \cite{nontriviallearn} mean a runtime of $2^{n-o(n)}$; prior work
had already connected non-trivial learning algorithms to circuit lower bounds \cite{klivans2013constructing, oliveira2016conspiracies}. 
The subexponential runtimes
we achieve in our work are significantly faster than the $2^{n-o(n)}$-time algorithms of \cite{nontriviallearn};
in addition, our algorithms work in the challenging agnostic setting, rather than just the PAC setting. 
On the other hand, the algorithms of \cite{nontriviallearn} apply to more general circuit classes than \LC. 

As mentioned previously, \cite{klivansservediodnfs} gave a $2^{\tO(n^{1/3})}$-time algorithm for PAC learning
polynomial size DNF formulas; their algorithm is based on a $\tO(n^{1/3})$ upper bound
on the \emph{threshold degree} of such formulas. 
In unpublished work, 
\cite{talpersonal} has observed that the argument in \cite[Theorem 4]{klivansservediodnfs} can be generalized to show
that for constant $d\geq 2$, any depth-$d$ \LC circuit has threshold
degree at most
 $\tO\bigl(n^{1-1/(3 \cdot 2^{d-3})}\bigr)$. This in turn yields a PAC learning algorithm for \LC running in time $\exp\bigl(\tO\bigl(n^{1-1/(3 \cdot 2^{d-3})}\bigr)\bigr)$. 
 Note that this is in the standard PAC model, not the agnostic PAC model.
 As mentioned in  \Cref{sec:intro}, prior to our work,
 no subexponential time algorithm was known for agnostically learning even \LCthree
 in subexponential time.

\subsection{Linear regression and the proof of \thm{learning}}
\label{s:techniques}

Our learning algorithm applies the well-known \emph{linear regression} framework for agnostic learning that was
introduced by \cite{kkms}. The algorithm of \cite{kkms} 
works whenever there is a ``small'' set of ``features'' $\F$ (where each feature is a function mapping $\B^n$ to $\R$) such that each concept in the concept class $\mathcal{C}$
can be approximated to error $\eps$ in the $\ell_{\infty}$ norm via a linear combination of the features in $\F$.
Roughly speaking, given a sufficiently large sample $S$ from an (unknown) distribution over $\B^n \times \B$, the algorithm
finds a linear combination $h$ of the features of $\F$ that minimizes the empirical $\ell_1$ loss, i.e., 
$h$ minimizes $\sum_{(x_i, b_i)\in S} |h(x_i)-b_i|$ among all linear combinations of features from $\F$. 
An (approximately) optimal $h$ can be found in time $\poly(\F)$ by solving a linear program of size 
 $\poly(|\F|, |S|)$. 
 
 \begin{lemma}[\cite{kkms}] \label{lem:kkms}
 Let $\F$ be a set of functions mapping $\B^n$ to $\R$, and assume that each $\phi_i \in \mathcal{F}$
 is efficiently computable, in the sense that for any $x \in \B^n$, $\phi_i(x)$ can be computed in time $\poly(n)$.
 Suppose that for every $c \in \mathcal{C}$, there exist
 coefficients $\alpha_i \in \mathbb{R}$ such that for all
 $x \in \B^n$, $|c(x) - \sum_{\phi_i \in \F} \alpha_i \cdot \phi_i(x)|\leq \eps$. 
 Then there is an algorithm that takes as input a sample $S$ of size $|S|=\poly(n, |\F|, 1/\eps, \log(1/\delta))$ from an unknown distribution $\D$,
 and in time $\poly(|S|)$ outputs a hypothesis $h$ such that, with
 probability at least $1-\delta$ over $S$, $\Pr_{(x, b) \sim \D}[h(x) \neq b] \leq \eps$. 
 \end{lemma}
 
 A  feature set $\mathcal{F}$ that is commonly used in applications of \lem{kkms}
 is the set of all monomials whose degree is at most some bound $d$.
 Indeed, an immediate corollary of \lem{kkms} is the following.
 
 \begin{corollary} \label{cor:thecor} \label{prop:agnosticfromadeg}
 Suppose that for every $c \in \mathcal{C}$, the $\eps$-approximate degree of $c$ is at most $d$. Then for every $\delta > 0$, there is an algorithm running in time $\poly(n^{d}, 1/\eps, \log(1/\delta))$ that agnostically
 learns $\mathcal{C}$ to error $\eps$ with respect to any (unknown) distribution $\D$ over $\B^n \times \B$.
 \end{corollary}

The best known algorithms for agnostically learning disjunctions and de Morgan formulas of linear size~\cite{kkms, Rei11} combine \cor{thecor}
with known approximate degree upper bounds for disjunctions and de Morgan formulas of bounded size. 
We use the same strategy: our results for agnostic learning (\thm{learning}) follow from combining
\cor{thecor} with our new approximate degree upper bounds. Specifically, \thm{approxdegLC} shows that the $\eps$-approximate degree of any $\LCd$ circuit is at most $\tilde{O}(n^{1-2^{-d}} \log^{2^{-d}}(1/\eps))$, yielding our new result for agnostically learning \LC circuits. \thm{upper} shows that \AC circuits of size $s$ have $\eps$-approximate degree $\tilde{O}(\sqrt{n} s^{1/2 - 2^{-d}}(\log(1/\eps))^{2^{-d}})$, giving our new result for learning sublinear-size \AC.

Furthermore, since our upper bound on the approximate degree of \LC circuits is nearly tight, new techniques will be needed to significantly surpass our results. In particular, new techniques will be needed to agnostically learn \emph{all} of \LC 
in subexponential time. 
\Cref{thm:lower} implies that if $\mathcal{F}$ is the set of all monomials of at most a given degree $d$, then one cannot
use \cor{thecor} to learn \LCd in time less than $2^{n^{1-2^{-\Omega(\sqrt{d})}}}$. However, 
 standard techniques \cite{patmat} automatically generalize the lower bound of \Cref{thm:lower} from the feature set of low-degree monomials
 to \emph{arbitrary feature sets}. Specifically, we obtain the following theorem.
 
 \label{sec:agnostics}
 \begin{theorem} \label{thm:genF}
 Let $\mathcal{C}=\LCd$, and let $\F^*$ denote the minimum size set of features such that
each $c \in \C$ can be approximated point-wise to error $1/3$ by a linear combination of the features in $\F$. 
Then  $|\F^*| \geq 2^{n^{1-2^{-\Omega(\sqrt{d})}}}$.
 \end{theorem}
 
For completeness, we provide the proof of \thm{genF} below.

\begin{proof}
 
For a matrix $F \in \B^{N \times N}$, 
the $\eps$-approximate rank of $F$, denoted $\text{rank}_{\eps}(F)$, is the least rank
of a matrix $A \in \R^{N \times N}$ such that $|A_{ij} - F_{ij}| \leq \eps$ for all $(i, j) \in [N] \times [N]$.
Sherstov's pattern matrix method \cite{patmat} allows one to translate in a black-box manner an approximate degree lower bound for a function $f$
into an approximate rank lower bound for a related matrix $F$, called the pattern matrix of $f$.

Specifically, invoking \Cref{thm:lower},
let $f$ be the function in \LCdmone satisfying $\adeg(f) \geq D$ for some $D=n^{1-2^{-\Omega(\sqrt{d})}}$. 
Viewing $F$
as a $2^{4n} \times 2^{4n}$ matrix in the natural way,
the pattern matrix method \cite[Theorem 8.1]{patmat} implies that
 the function 
$F \colon \B^{4n}\times \B^{4n} \to \B$ given by
$F(x, y) = f\left(\dots, \vee_{j=1}^4 \left(x_{i, j} \wedge y_{i, j}\right) \dots\right)$
satisfies
 \begin{equation}
 \label{eq:nomore} \text{rank}_{1/3}(F) \geq 2^{\Omega(D)},\end{equation}
 where the expression $\text{rank}_{1/3}(F)$ views $F$ as a $2^{4n} \times 2^{4n}$ matrix.
 
Let  $\F^*$ be a feature set satisfying
the hypothesis of \Cref{thm:genF}, i.e., for every function $c \colon \B^{4n} \to \B$ in \LCd, there exist constants $\alpha_1, \dots, \alpha_{|\F|}$ 
such that 
\begin{equation}
|c(x)-\sum_{\phi_j \in \F} \alpha_j \phi_j(x)| \leq 1/3
\end{equation}
for all $x \in \B^{4n}$. We claim
that this implies that \begin{equation}\label{Really} \text{rank}_{1/3}(F) \leq |\F^*|. \end{equation}
\Cref{thm:genF} then follows by combining Equation \eqref{Really} with Equation \eqref{eq:nomore}.
 
Let us view each row $i$ of $F$ as a function $F_i$ mapping $\B^{4n} \to \B$. Then
clearly, if $f$ is in \LCdmone, each row $F_i$ is in \LCd. 
Hence, there exist constants $\alpha_{i,1}, \dots, \alpha_{i,|\F|}$ such that
\begin{equation} \label{anotherone} |F_i(x)-\sum_{\phi_j \in \F} \alpha_{i, j}  \cdot \phi_j(x)| \leq 1/3 \text{ for all } x \in \B^{4n}. \end{equation}
 
Let $M$ denote the $2^{4n} \times |\F|$ matrix whose $(i, j)$'th entry is
$\alpha_{i, j}$. And let $R$ denote that $|\F| \times 2^{4n}$ matrix
whose $(j, x)$'th entry is $\phi_{j}(x)$, where we associate $x$ with an input in $\B^{4n}$. 
Then Equation \eqref{anotherone} implies that $|M \cdot R - F_{ij}| \leq 1/3$ for all $(i, j) \in [2^{4n}] \times [2^{4n}]$.
Since $M \cdot R$ is a matrix of rank at most $|\F|$, Equation \eqref{Really} follows.
\end{proof}

\section{Circuit Lower Bounds (Proof of \Cref{apptheorem})}
\label{s:app}
In this section, we view Boolean functions as mapping domain $\{-1, 1\}^n$ to $\{-1, 1\}$. Recall that $$\IP(x, y) = \oplus_{i=1}^n (x_i \wedge y_i)$$ denotes the Boolean inner product on $2n$ bits.
As a warmup, we start by establishing a worst-case version of \Cref{apptheorem}.

\begin{proposition} \label{worstcaseprop} The Inner Product function cannot be computed by any depth-$(d+1)$ \ACmodtwo\ circuit of size $\tOmega\bigl(n^{1/(1-2^{-d})}\bigr)$. 
\end{proposition}
\begin{proof} \Cref{thm:approxdegLC} shows that any depth-$d$ \AC circuit of size $s \geq n$ on $n$ inputs  has approximate degree at most $D= \tilde{O}(s^{1-2^{-d}})$. Clearly, the approximating
polynomial has at most $\binom{s}{\leq D} \leq s^D$ many monomials. 

From this, one can conclude that any depth-$(d+1)$ \ACmodtwo\ circuit $\C$ on $n$ inputs of size $s \geq n$ can be approximated by a polynomial $p$ over $\{-1, 1\}^n$ with at most $\binom{s}{D}$ many monomials. 
To see why, let us write $\C(x, y)= \C'(h_1(x, y), \dots, h_N(x, y))$, where $N \leq s$, $\C'$ is an AC$^0$ circuit of depth $d$ and size at most $s$, and each $h_i$ is a parity function.  
Since $\C'$ is an \AC circuit of depth $d$ and size  at most $s$ on $N \leq s$ inputs, it has approximate degree at most $D$. Accordingly, let $q$ be a polynomial
of degree at most $D$ that point-wise approximates $\C'$ to error at most $1/3$.
Now obtain $p$ by replacing the $i$'th input to $q$ with the corresponding parity gate, namely $h_i$, of $\C$. This yields a polynomial $p$ that point-wise approximates $\C$ to error at most $1/3$, i.e., $|p(x, y) - \C(x, y)| \leq 1/3$ for all $(x, y) \in \{-1, 1\}^n \times \{-1, 1\}^n$. Since $q$ is defined over domain $\{-1, 1\}^N$, replacing any number of inputs to $q$ with parity functions preserves the number of monomials of $q$. 

On the other hand, it is known that that any polynomial $p$ over $\{-1, 1\}^n \times \{-1, 1\}^n$ that point-wise approximates the Inner Product function to any error strictly less than 1 requires $2^{\Omega(n)}$ many monomials \cite{bruck1992polynomial}.

Combining the above two facts means that $s^D$ must be at least $2^{\Omega(n)}$, which means that $s$ must be at least $\tilde{\Omega}(n^{1/(1-2^{-d})})$.  
\end{proof}

We now prove \Cref{apptheorem}, restated here for convenience. 

\apptheorem*

 \paragraph{Proof Outline.}
 The proof follows a similar outline to \Cref{worstcaseprop}, but builds on an observation of Tal  \cite[Lemma 4.2]{tal2016bipartite}.
Roughly, Lemma 4.2 of \cite{tal2016bipartite} shows that bipartite de Morgan formulas of size $s$ cannot compute the Inner Product function 
on more than a $1/2 + n^{- \log n}$ fraction of inputs unless they have size at least roughly $n^2$. 
The only property of de Morgan formulas of size $\ll n^2$ that Tal uses is that they have sublinear approximate degree.

Similarly, \Cref{thm:approxdegLC}  shows that an \AC circuit of size $s$ and depth $d$ on $n$ inputs, for which  $n \leq s \ll n^{1/(1-2^{-d})}$, has  sublinear approximate degree.

Any parity function is an example of a bipartite function of size $O(1)$, meaning that the parity function applied to some subset of an input $(x, y) \in \{-1, 1\}^n \times \{-1, 1\}^n$ is computable by a constant-sized circuit with leaves computing a function of only $x$ or $y$. 
Hence, Tal's argument applies with cosmetic changes not only to sub-quadratic size bipartite de Morgan formulas, but also to \ACmodtwo\  circuits of size $s \ll n^{1/(1-2^{-d})}$. 

We remark that the entire argument (and hence the lower bound of \Cref{apptheorem} itself) applies not only to \ACmodtwo\ circuits, but more generally to depth-$d$ \AC circuits augmented with a layer of low-communication
gates above the inputs; we omit this extension for brevity.

\begin{proof}[Proof of \Cref{apptheorem}, closely following the proof of  Lemma 4.2 of \cite{tal2016bipartite}]
Let $\C \colon \{-1, 1\}^{2n} \to \{-1, 1\}$ be an \ACmodtwo\ circuit of depth $(d+1)$ and size $s \geq n$, and let
$$q = \Pr_{x, y \in \{-1, 1\}^n}[\C(x, y) = \IP(x, y)].$$ Suppose that $q \geq 1/2 + \eps$. Our goal is to show that $s$ must be large, even for negligible values of $\eps$.

Let $N \leq s$ denote the number of parity gates in $\C$, with the $i$th parity gate denoted by $h_i(x) \colon \{-1, 1\}^n \to \{-1, 1\}$. 
Then we may write $\C(x, y) = \C'(h_1(x, y), \dots, h_N(x, y))$, where
$\C'$ is an \AC
circuit on at most $s$ inputs, of depth $d$ and size at most $s$. 
By \Cref{thm:approxdegLC}, there exists a polynomial $p$ of degree at most $D \leq \tilde{O}\left(s^{1-2^{-d}} \log^{2^{-d}}(1/\eps)\right)$ such that,
for all $w \in \{-1, 1\}^N$, $|p(w)-\C'(w)| \leq \eps$.

Next, we show that under the uniform distribution, the function $\IP(x, y)$
correlates well with $p(h_1(x), \dots, h_N(x))$. 
We decompose the expectation $\mathbf{E}_{x, y \in \{-1, 1\}^n}[p(x, y) \cdot \IP(x, y)]$ according to whether or not $\IP(x, y) = \C(x, y)$:
\begin{align} 
\mathbf{E}_{x, y \in \{-1, 1\}^n}[p( h_1(x, y), \dots, h_N(x, y)) \cdot \IP(x, y)] &= \notag   \\
\mathbf{E}_{x, y \in \{-1, 1\}^n}[p( h_1(x, y), \dots, h_N(x, y)) \cdot \IP(x, y)& | \IP(x, y) = \C(x, y)] \cdot \Pr[\IP(x, y) = \C(x, y)] + \notag \\
\mathbf{E}_{x, y \in \{-1, 1\}^n}[p( h_1(x, y), \dots, h_N(x, y)) \cdot \IP(x, y)& | \IP(x, y) \neq \C(x, y)] \cdot \Pr[\IP(x, y) \neq \C(x, y)] \notag \\
 &\geq (1 - \eps) \cdot q + (-1 - \eps) \cdot (1 - q) \notag \\
 &=  2q - 1 - \eps  \geq 2 \cdot (1/2 + \eps) - 1 - \eps = \eps. \label{eq:corrbound}
\end{align}

Next, we write $p(z)$ as a multi-linear polynomial:
$p(z) = \sum_{S \subseteq [N], |S| \leq D}  \hat{p}(S) \cdot \prod_{i \in S} z_i$.
Since $\hat{p}(S) = \mathbf{E}_{z \in \{-1,1\}^N} [p(z) \cdot \prod_{i \in S} z_i]$, 
we have that $|\hat{p}(S)| \leq 1 + \eps$ for every $S$. Note that there are at most
$\binom{N}{\leq D}$ monomials in $p$. Invoking Inequality \eqref{eq:corrbound}, we have:

\begin{align*} \eps &\leq \mathbf{E}_{x, y \in \{-1, 1\}^n}\left[p(h_1(x, y), \dots , h_N(x, y)) \cdot \IP(x, y)\right]\\
&=\mathbf{E}_{x, y \in \{-1, 1\}^n}\left[\sum_{S \subseteq [N], |S| \leq D} \hat{p}(S) \prod_{i \in S} h_i(x, y)  \cdot \IP(x, y)\right]\\
&=\sum_{S \subseteq [N], |S|\leq D} \hat{p}(S) \cdot \mathbf{E}_{x, y \in \{-1, 1\}^n}\left[\prod_{i \in S} h_i(x, y)  \cdot \IP(x, y)\right]\\
&\leq \sum_{S \subseteq [N], |S|\leq D} (1+\eps) \left| \mathbf{E}_{x, y \in \{-1, 1\}^n}\left[\prod_{i \in S} h_i(x, y)  \cdot \IP(x, y)\right]\right|.
\end{align*} 

Hence there must exist a set $S \subseteq [N]$ with size at most $D$ such that
$$\left| \mathbf{E}_{x, y \in \{-1, 1\}^n}\left[\prod_{i \in S} h_i(x, y)  \cdot \IP(x, y)\right]\right| \geq \frac{\eps}{\binom{N}{\leq D}\left(1+\eps\right)} \geq (\eps/2) \cdot s^{-D}  \geq \exp\left({\tilde{O}(-s^{1-2^{-d}} \log^{2^{-d}}(1/\eps))}\right).$$
 
It is well-known that $\IP$ is $2^{-\Omega(n)}$ correlated with any parity function $h_i$ (indeed,
$\IP$ on $2n$ bits is a \emph{bent} function, meaning that all its Fourier coefficients have magnitude $2^{-n}$,
and hence its correlation with any parity is at most $2^{-n}$). 
We conclude that  $$s^{1-2^{-d}} \log^{2^{-d}}(1/\eps)  \geq \tilde{\Omega}(n).$$
The theorem is an immediate consequence of this inequality.
\end{proof}

\section*{Acknowledgements}

We thank Nikhil Mande, Ronald de Wolf, and Shuchen Zhu for comments on earlier drafts of this paper. R.K. thanks Luke Schaeffer for comments on the proof of \thm{quantum}.

\bibliographystyle{alphaurl}
\phantomsection\addcontentsline{toc}{section}{References} 
\renewcommand{\UrlFont}{\ttfamily\small}
\let\oldpath\path
\renewcommand{\path}[1]{\small\oldpath{#1}}
\bibliography{sharedinputs}

\appendix

\section{Depth-vs-Size Upper Bounds for Approximate Majorities}
In this section, we prove a result alluded to in \Cref{s:discussion},
namely an upper bound on the size of small-depth circuits computing approximate majorities.

For $0 < p < q < 1$, we use $\AMAJ_{n, p, q}$ to refer to any total function $f$ on $n$ 
bits satisfying the following two properties:
\begin{itemize}
\item $f(x) = 0$ for all $x$ of Hamming weight at most $pn$.
\item $f(x) = 1$ for all $x$ of Hamming weight at least $qn$.
\end{itemize}

We also say that such an $f$ computes $\AMAJ_{n, p, q}$.

\begin{theorem} \label{appthm}
For any constant $\delta > 0$, there are positive constants $0 < p < 1/2 < q < 1$ such that the following holds.
There is a total function $f$ computing $\AMAJ_{n, p, q}$ such that $f$ is itself
computable by a depth-$(2d+3)$ circuit of size $O_{d, \delta}(n^{1+2^{-d} + \delta})$. Here, the $O_{d, \delta}$ notation hides a leading factor that depends only on $d$ and $1/\delta$.
\end{theorem}
\begin{proof} 
A careful application of well-known arguments yields the following claim, whose proof we defer to the end of this section. 
\begin{claim} For any constant $\delta > 0$, there exist constants $0 < p_0 < 1/2 < q_0 < 1$ such that the following holds.
There is a monotone depth-3 \AC circuit $C_0$ of size $O(m^{2+\delta})$ computing $\AMAJ_{m, p_0, q_0}$. The top and bottom layers of gates of $C_0$ are $\AND$
gates. \label{depth3claim}
\end{claim}

To prove \Cref{appthm}, we need to use $C_0$ to construct deeper but smaller circuits that also compute approximate majorities. 
More generally, we establish the following iterative transformation that takes any circuit $C_i$
computing an approximate majority function and turns it into a smaller and only slightly deeper circuit $C_{i+1}$ that also computes an
approximate majority function.\footnote{Not coincidentally, this iterative transformation to reduce circuit size
at the expense of depth is reminiscent of the transformation
used to prove the approximate degree lower bound for linear-size circuits given in  \Cref{thm:lower}.}
 
\begin{lemma} \label{thekeylemmaman}
Let $d>0$ and $\delta >0$ be fixed constants, and let $p_0$ and $q_0$ be the associated constants from \Cref{depth3claim}. Suppose that 
there exists a family of monotone circuits $\C_i$ with the following properties.
\begin{itemize}
\item There exist constants $p_i$ and $q_i$ satisfying $0 < p_i \leq p_0 < 1/2 <  q_0 \leq q_i < 1$, such that for all input sizes $m$, $\C_i$ contains a circuit computing  $\AMAJ_{m, p_i, q_i}$.
\item  Each circuit in $\C_i$ has depth $2i+3$, and the top and bottom layers consist of $\AND$ gates. 
\item  There is a constant $k_i >0$ such that the circuit in $\C_i$ defined over inputs of size $m$ has size at most $O_{d, p_i, q_i}(m^{k_i + \delta})$, 
where the $O_{d, p_i, q_i}$ notation hides factors depending only on $d$, $p_i$, and $q_i$. 
\end{itemize} 

Then there exists a family of monotone circuits $\C_{i+1}$  
satisfying the above three properties, with $p_{i+1} = (1-1/(10d)) p_i$,  $q_{i+1} = 1-(1-q_i)(1-1/(10d))$, and $k_{i+1} \leq (1 + k_{i})/2$. 
\end{lemma}

\begin{proof}

Let $C_i$ be the assumed circuit from family $\C_i$ on $m$ inputs computing $\AMAJ_{m, p_i, q_i}$.
Let $$n=m^2$$ and 
$$M=700 d^2 (1/p^2_i + 1/q^2_i) m.$$ 
Consider generating a circuit $C_{i+1}$ on $n$ inputs via the following random process. $C_{i+1}$ will have the form \begin{equation}
\label{compeq} \AMAJ_{M, p_i, q_i} (\AMAJ_{m, p_i, q_i}, \dots, \AMAJ_{m, p_i, q_i}),\end{equation} Here, $p_i$ and $q_i$ are fixed constants as per the statement of the lemma, and each of the bottom $\AMAJ$ circuits are evaluated on a randomly chosen (size-$m$) subset of the $n$ inputs of $C_{i+1}$. Since $p_i \leq p_0$ and $q_0 \leq q_i$, we may use a circuit $C_0$ from family $\C_0$ (as per \Cref{depth3claim}) to compute the outer function $\AMAJ_{M, p_i, q_i} $. We use the circuit $C_i$ from family $\C_i$
to compute each copy of the inner function $\AMAJ_{m, p_i, q_i}$. 

In summary, we generate the circuit $C_{i+1}$ to be the composition $C_0 \circ C_i$, but where each copy of $C_i$ is evaluated over a randomly chosen (size-$m$) subset of the $n$ inputs of $C_{i+1}$ (i.e., $C_{i+1}$ is a shared-input composition of $C_0$ and $C_i$).

We claim that with strictly positive probability, this circuit $C_{i+1}$ computes $\AMAJ_{n, p_{i+1}, q_{i+1}}$. 
To see this, first fix an input $x$ with Hamming weight at most $p_{i+1} \cdot n$, so that the expected
number of $1$-inputs to any bottom $\AMAJ_{m, p_i, q_i}$ circuit is at most $ \mu := p_{i+1}  \cdot m$. Note that  $p_i \cdot m > (1+1/(10d)) \mu$. 
If any $\AMAJ_{m, p_i, q_i}$ circuit ``makes an error'' on $x$ (i.e., evaluates to $1$ on $x$), then at least $p_i \cdot m > (1+1/(10d)) \cdot \mu$ of the randomly chosen inputs to the gate are 1.
By a Chernoff bound, for each of the bottom 
$\AMAJ_{m, p_i, q_i}$ gates, this happens on input $x$ with probability at most $\exp(-\mu/(3(10d)^2)) \leq \exp(-\mu/(300d^2)) \leq \exp(-p_i m/(600d^2))$. 

The probability that more than $(700 d^2/p_i) m \leq p_i \cdot M$ of these circuits makes an error is at most $2^{M} \cdot (\exp(-p_im/(600d^2)))^{(700d^2/p_i)m } \ll \exp(-m^2)$. Thus, with probability at least $1-\exp(-m^2)$, the circuit $C_{i+1}$ outputs 0 on input $x$. 

An analogous argument holds for inputs $x$ with Hamming weight at least $q_{i+1} \cdot n$,
so by a union bound over all at most the $2^n$ inputs to $C_{i+1}$ with Hamming weight at most $p_{i+1} \cdot n$ or at least $q_{i+1} \cdot n$, with strictly positive probability $C_{i+1}$ computes $\AMAJ_{n, p_{i+1}, q_{i+1}}$.

The circuit $C_{i+1}$ has $m^2$ inputs and has size at most $$O( M^{2+\delta}) + O(M \cdot m^{k_i + \delta}) = O_{d, p_i, q_i}(n^{1+\delta/2} + m^{1+k_i + \delta}) = O_{d, p_i, q_i}(n^{k_{i+1} + \delta/2}),$$
where recall that $k_{i+1} = (1 + k_{i})/2$. 

Equation \eqref{compeq} implies that the top and bottom layers of $C_{i+1}$ consist of $\AND$ gates, with $C_{i+1}$ inheriting this property directly
from $C_i$ and $C_0$. Moreover, by collapsing the bottom layer of $C_0$ with the top layer of each copy of $C_i$ (which is possible because $C_0$ is monotone),
we find that the depth of $C_{i+1}$ is at most most $3 + (2i+3) - 1 = 2 (i+1)+3$.  This completes the proof of the lemma.
\end{proof}

Let $p_0, q_0$ be as in \Cref{depth3claim}, and let $p=p_0/e$ and $q=1-(1-q_0)/e$. 
\Cref{appthm} follows by iteratively applying \Cref{thekeylemmaman} $d$ times (starting with $i=0$; the assumptions of the lemma are satisfied for this value of $i$
by \Cref{depth3claim}) to conclude that $\AMAJ_{n, p, q}$
is computable by a circuit of depth $2d+3$ and size $O_d(n^{1+2^{-d} + \delta})$.
\end{proof}

\begin{proof}[Proof of \Cref{depth3claim}]
The main idea of the (probabilistic) construction is to have an $\AND$-$\OR$-$\AND$ circuit $C$, where the top $\AND$
gate has fan-in $t_1 := m$, the middle layer (of $\OR$ gates) all have fan-in $t_2 := m^{1+\delta}$, and the bottom layer of $\AND$ gates 
all have fan-in $t_3=\log_2(m)$.  Each bottom $\AND$ gate is connected to $t_3$ randomly chosen inputs. 

Let $p$ be any constant less than $1/2^{1+\delta}$, and $q = 1/2^{\delta}$. These choices ensure that $p^{\log_2(m)} < 1/(2m^{1+\delta})$ and $q^{\log_2(m)} > 1/m^{\delta}$. 
We now show that with positive probability, $C$ computes $\AMAJ_{m, p, q}$. 

Consider any $m$-bit input $x$ with Hamming weight at most $p \cdot m$. 
Then for any fixed $\AND$ gate at the bottom layer of $C$, the probability the $\AND$
gate evaluates to $1$ is at most $p^{t_3} < 1/(2m^{1+\delta})$. 
By a union bound, this implies that for any fixed $\OR$ gate at the middle layer of $C$, the probability the $\OR$ gate outputs 1 on $x$ 
is at most $t_2 \cdot 1/(2m^{1+\delta}) \leq 1/2$. This implies that the probability the top $\AND$ gate outputs $1$ on $x$ is at most $1/2^{t_1} = 2^{-m}$.

Now consider any $m$-bit input $x$ with Hamming weight at least $q \cdot m$. 
Then for any fixed $\AND$ gate at the bottom layer of $C$, the probability the $\AND$
gate evaluates to $1$ is at least $q^{t_3} > 1/m^{\delta}$. 
This implies that for any fixed $\OR$ gate at the middle layer of $C$, the probability the $\OR$ gate outputs 1 on $x$ 
is least $1-(1-1/m^{\delta})^{t_2} \geq 1-e^{-m} \geq 1-1/(m2^{m})$. This implies that the probability the top $\AND$ gate outputs $1$ on $x$ is at least $1-2^{-m}$.

By a union bound over all  the at most $2^m$ inputs $x$ to $C$, we conclude that with positive probability $C$ computes $\AMAJ_{m, p, q}$.
\end{proof}

\label{app:amaj}
\end{document}